\documentclass[11pt]{article}
\usepackage{amsfonts,amsmath,bm,xspace,amsthm,fullpage,amssymb}
\usepackage{multirow,graphicx, algorithm,algorithmic,rotating}
\usepackage{url,cite}
\usepackage{hyperref}
\usepackage[small,bf]{caption}
\theoremstyle{definition}
\newtheorem{theorem}{Theorem}

\newtheorem{lemma}{Lemma}
\newtheorem{corollary}{Corollary}
\newtheorem{proposition}{Proposition}

\newtheorem{definition}{Definition}

\newtheorem*{note}{Note}

\numberwithin{equation}{section}

\DeclareMathOperator{\conv}{conv}

\DeclareMathOperator{\supp}{supp}

\DeclareMathOperator{\cone}{cone}

\DeclareMathOperator{\E}{\mathbb{E}} 

\newcommand{\A}{\mathcal{A}} 
\newcommand{\C}{\mathbb{C}}
\newcommand{\R}{\mathbb{R}}
\newcommand{\vnorm}[1]{\lVert#1\rVert} 
\newcommand{\vabs}[2]{\langle#1, #2\rangle}

\renewcommand{\Re}{\mathop{\mathrm{Re}}}

\newcommand{\minimize}{\mathop{\operatorname{minimize}}}
\newcommand{\maximize}{\mathop{\operatorname{maximize}}}

\newcommand{\eq}[1]{(\ref{eq:#1})}

\title{\bf \large Atomic norm denoising with applications to line spectral estimation\footnote{A preliminary version of this work appeared in the Proceedings of the 49th Annual Allerton Conference in 2011.}}

\author{Badri Narayan Bhaskar$^\dagger,$ Gongguo Tang$^\dagger,$ and Benjamin Recht$^\sharp$\\
$^\dagger$Department of Electrical and Computer Engineering\\
$^\sharp$Department of Computer Sciences\\
University of Wisconsin-Madison
}

\date{April 2012.  Last Revised Feb.~2013.}

\begin{document}

\maketitle

\vspace{-0.3in}

\begin{abstract}
Motivated by recent work on atomic norms in inverse problems, we propose a new
approach to line spectral estimation that provides theoretical guarantees for
the mean-squared-error (MSE) performance in the presence of noise and without 
knowledge of the model order. We propose an abstract theory of denoising with
atomic norms and specialize this theory to provide a convex optimization
problem for estimating the frequencies and phases of a mixture of complex
exponentials. We show that the
associated convex optimization problem can be solved in polynomial time via
semidefinite programming (SDP). We also show that the SDP can be approximated
by an $\ell_1$-regularized least-squares problem that achieves nearly the same error rate as the SDP but can
scale to much larger problems. We compare both SDP and $\ell_1$-based approaches with
classical line spectral analysis methods and demonstrate that the SDP outperforms the
$\ell_1$ optimization which outperforms MUSIC, Cadzow's, and Matrix Pencil approaches in
terms of MSE over a wide range of signal-to-noise
ratios. \end{abstract}

\section{Introduction}

Extracting the frequencies and relative phases of a superposition of complex
exponentials from a small number of noisy time samples is a foundational
problem in statistical signal processing. These \emph{line spectral estimation}
problems arise in a variety of applications, including the direction of arrival
estimation in radar target identification~\cite{radar}, sensor array signal
processing~\cite{arrays} and imaging systems~\cite{imaging} and also 
underlies techniques in ultra wideband channel
estimation~\cite{uwb}, spectroscopy~\cite{nmr}, molecular
dynamics~\cite{andrade2012}, and power electronics~\cite{power}.

While polynomial interpolation using Prony's technique can estimate the
frequency content of a signal \emph{exactly} from as few as $2k$ samples if
there are $k$ frequencies, Prony's method is inherently unstable due to
sensitivity of polynomial root finding. Several methods have been proposed to
provide more robust polynomial interpolation~\cite{music,esprit,hua02} (for an
extensive bibliography on the subject, see ~\cite{stoica93}), and these
techniques achieve excellent noise performance in moderate noise. However, the
denoising performance is often sensitive to the model order estimated, and
theoretical guarantees for these methods are all asymptotic with no finite
sample error bounds. Motivated by recent work on atomic norms~\cite{crpw}, we
propose a convex relaxation approach to denoise a mixture of complex
exponentials, with theoretical guarantees of noise robustness and a better
empirical performance than previous subspace based approaches.

Our first contribution is an abstract theory of denoising with atomic norms.
Atomic norms provide a natural convex penalty function for discouraging
specialized notions of complexity. These norms generalize the $\ell_1$ norm for
sparse vector estimation~\cite{candes06} and the nuclear norm for low-rank
matrix reconstruction~\cite{Recht10,CandesRecht09}. We show a unified
approach to denoising with the atomic norm that provides a standard approach to
computing low mean-squared-error estimates. We show how certain Gaussian
statistics and geometrical quantities of particular atomic norms are sufficient
to bound estimation rates with these penalty functions. Our approach is
essentially a generalization of the Lasso~\cite{tibshirani96,chen01} to infinite
dictionaries.
 
Specializing these denoising results to the line spectral estimation,
we provide mean-squared-error estimates for denoising line spectra with the
atomic norm. The denoising algorithm amounts to soft thresholding the noise
corrupted measurements in the atomic norm and we thus refer to the problem as
\emph{Atomic norm Soft Thresholding} (AST). We show, via an appeal to the
theory of positive polynomials, that AST can be solved using semidefinite
programming (SDP)~\cite{Megretski03}, and we provide a reasonably fast method
for solving this SDP via the Alternating Direction Method of Multipliers
(ADMM)~\cite{BertsekasParallelBook,admm2011}. Our ADMM implementation can 
solve instances with a thousand observations in a few minutes.

While the SDP based AST algorithm can be thought of as solving an infinite
dimensional Lasso problem, the computational complexity can be prohibitive for
very large instances. To compensate, we show that solving the Lasso problem on
an oversampled grid of frequencies approximates the solution of the atomic norm
minimization problem to a resolution sufficiently high to guarantee excellent
mean-squared error (MSE). The gridded problem reduces to the Lasso, and by
leveraging the Fast Fourier Transform (FFT), can be rapidly solved with freely
available software such as SpaRSA~\cite{wright09}. A Lasso problem with
thousands of observations can be solved in under a second using Matlab on a
laptop. The prediction error and the localization accuracy for line spectral
estimation both increase as the oversampling factor increases, even if the
actual set of frequencies in the line spectral signal are off the Lasso grid.

We compare and contrast our algorithms, AST and Lasso, with classical line
spectral algorithms including MUSIC~\cite{music},and Cadzow's ~\cite{cadzow02} 
and  Matrix Pencil~\cite{hua02} methods . Our experiments
indicate that both AST and the Lasso approximation outperform classical methods
in low SNR even when we provide the exact model order to the classical
approaches. Moreover, AST has the same complexity as Cadzow's method,
alternating between a least-squares step and an eigenvalue thresholding step.
The discretized Lasso-based algorithm has even lower computational complexity,
consisting of iterations based upon the FFT and simple
linear time soft-thresholding.

\subsection{Outline and summary of results} We describe here our approach to a
general sparse denoising problem and later specialize these results to line
spectral estimation. The denoising problem is obtaining an estimate $\hat{x}$
of the signal $x^\star$ from $y = x^\star + w$, where $w$ is additive noise. We
make the structural assumption that $x^\star$ is a sparse non-negative
combination of points from an arbitrary, possibly infinite set $\A \subset
\C^n$. This assumption is very expressive and generalizes many notions of
sparsity~\cite{crpw}. The atomic norm
$\vnorm{\cdot}_\A$, introduced in \cite{crpw}, is a penalty function specially
catered to the structure of $\A$ as we shall examine in depth in next section,
and is defined as: \begin{equation*} \vnorm{x}_\A = \inf\left\{t > 0 ~\middle|~
x \in t \conv(\A) \right\}. \end{equation*} where $\conv(\A)$ is the convex
hull of points in $\A.$ We analyze the denoising performance of an estimate
that uses the atomic norm to encourage sparsity in $\A$.

\paragraph*{Atomic norm denoising.} In Section \ref{sec:abstract-denoising}, we characterize the performance of the
estimate $\hat{x}$ that solves
\begin{equation}
\label{AST}\mathop{\textrm{minimize}}_x \frac{1}{2} \vnorm{x - y}_2^2 + \tau \vnorm{x}_\A.
\end{equation}
where $\tau$ is an appropriately chosen regularization parameter. We provide an
upper bound on the MSE when the noise statistics are
known. Before we state the theorem, we note that the dual norm
$\vnorm{\cdot}_\A^*$, corresponding to the atomic norm, is given by
\[
 \vnorm{z}_{\A}^* = \sup_{a \in \A}{\vabs{z}{a}},
\]
where $\vabs{x}{z} = \Re(z^*x)$ denotes the real inner product.  

\begin{theorem}
 \label{cor:expected-mse}

Suppose we observe the signal $y = x^\star + w$ where $x^\star \in \C^n$ is a
sparse nonnegative combination of points in $\A$. The estimate $\hat{x}$ of
$x^\star$ given by the solution of the atomic soft thresholding problem
\eqref{AST} with $\tau \geq \E \vnorm{w}_\A^*$ has the expected (per-element)
MSE
\[ 
\frac{1}{n}\E \vnorm{\hat{x} - x^\star}_2^2 \leq \frac{\tau}{n}\vnorm{x^\star}_\A
\]

\end{theorem}

This theorem implies that when $\E \vnorm{w}_\A^*$ is $o(n)$, the estimate
$\hat{x}$ is consistent.
 
\paragraph*{Choosing the regularization parameter.} Our lower bound on $\tau$ is in terms of the expected dual norm of the noise
process $w$, equal to
 \[
 	\E \vnorm{w}_\A^* = \mathbb{E}[\sup_{a\in \A} {\vabs{w}{a}}].
 \]
That is, the optimal $\tau$ and achievable MSE can be estimated
by studying the extremal values of the stochastic process indexed by the atomic
set $\A$.
 
\paragraph*{Denoising line spectral signals.} After establishing the abstract theory, we specialize the results of the
abstract denoising problem to line spectral estimation in
Section~\ref{sec:denoise-trig-moments}. Consider the continuous time signal $x^\star(t), t\in \R$ with a line
spectrum composed of $k$ unknown frequencies $\omega_1^\star, \ldots,
\omega_k^\star$ bandlimited to $[-W,W]$. Then the Nyquist samples of the signal
are given by
\begin{equation}
\label{eq:tru-moment}
 x^\star_m := x^\star\left(\tfrac{m}{2W}\right) = \sum_{l=1}^k c_l^\star e^{i 2 \pi m f_l^\star}, m = 0, \ldots, n-1
\end{equation}
where $c_1^\star, \ldots, c_k^\star$ are unknown \emph{complex} coefficients
and $f_l^\star = \tfrac{\omega_l^\star}{2W}$ for $l = 1, \ldots, k$ are the
normalized frequencies. So, the vector $x^\star = [x^\star_0 ~ \cdots ~
x^\star_{n-1}]^T \in \C^n$ can be written as a non-negative
linear combination of $k$ points from the set of atoms
{\small
\[
\A = \left\{e^{i 2\pi \phi}[1 ~ e^{i2\pi f} ~ \cdots ~ e^{i2\pi (n-1) f}]^T,\, f \in [0,1], \phi \in [0,1] \right\}.
\]} 
The set $\A$ can be viewed as an infinite dictionary indexed by the
continuously varying parameters $f$ and $\phi$. When the number of observations,
$n$, is much greater than $k$, $x^\star$ is $k$-sparse and thus line spectral estimation in the
presence of noise can be thought of as a sparse approximation
problem. The regularization parameter for the strongest guarantee
in Theorem \ref{cor:expected-mse} is given in terms of the expected dual norm
of the noise and can be explicitly computed for many noise models. For example,
when the noise is Gaussian, we have the following theorem for the
MSE:

\begin{theorem}
\label{thm:expmsels}

Assume $x^\star \in \C^n$ is given by $x_m^\star = \sum_{l=1}^k{c_l^\star
e^{i2\pi m f_l^\star}}$ for some unknown complex numbers $c_1^\star, \ldots,
c_k^\star$, unknown normalized frequencies $f_1^\star, \ldots, f_k^\star \in
[0,1]$ and $w \in \mathcal{N}(0,\sigma^2 I_n)$. Then the estimate $\hat{x}$ of
$x^\star$ obtained from $y=x^\star+w$ given by the solution of atomic soft
thresholding problem \eqref{AST} with $\tau = \sigma \sqrt{n \log(n)}$ has the
asymptotic MSE
\belowdisplayskip=-10pt
\[
\frac{1}{n} \E \vnorm{\hat{x} - x^\star}_2^2 \lesssim
	\sigma\sqrt{\frac{\log(n)}{n}}\sum_{l=1}^k |c_l^\star|.
\]
 \end{theorem}
 
It is instructive to compare this to the trivial estimator $\hat{x} = y$ which
has a per-element MSE of $\sigma^2$. In contrast, Theorem 2 guarantees that AST
produces a consistent estimate when $k = o(\sqrt{n/\log(n)})$.

\paragraph*{Computational methods.} We show in Section \ref{sec:sdp-ast} that \eqref{AST} for line spectral
estimation can be reformulated as a semidefinite program and can be solved on
moderately sized problems via semidefinite programming. We also show that we
get similar performance by discretizing the problem and solving a Lasso
problem on a grid of a large number of points using standard $\ell_1$ minimization software. Our discretization results justify the
success of Lasso for frequency estimation problems (see for instance,
\cite{chen98spectrum,malioutov05,bourguignon2007irregular}),
even though many of the common theoretical tools for compressed sensing do not
apply in this context. In particular, our measurement matrix does not obey RIP
or incoherence bounds that are commonly used. Nonetheless, we are able to
determine the correct value of the regularization parameter, derive estimates on the MSE, and obtain excellent denoising in
practice.

\paragraph*{Localizing the frequencies using the dual problem.} The atomic formulation not only offers a way to denoise the line spectral
signal, but also provides an efficient frequency localization method. After we
obtain the signal estimate $\hat{x}$ by solving \eqref{AST}, we also obtain
the solution $\hat{z}$ to the dual problem as $\hat{z} = y - \hat{x}$. As we
shall see in Corollary 1, the dual solution $\hat{z}$ both certifies the optimality of $\hat{x}$ and reveals the composing atoms of $\hat{x}$. For line spectral estimation, this provides an alternative to polynomial interpolation for localizing the constituent frequencies.

Indeed, when there is no noise, Cand\'es and Fernandez-Granda showed the dual
solution recovers these frequencies exactly under mild technical
conditions~\cite{CandesGranda}. This frequency localization technique is later
extended in ~\cite{offgrid2012} to the random undersampling case to yield a
compressive sensing scheme that is robust to basis mismatch. When there is
noise, numerical simulations show that the atomic norm minimization problem
\eqref{AST} gives approximate frequency localization.

\paragraph*{Experimental results.} A number of Prony-like techniques have been devised that are able to achieve
excellent denoising and frequency localization even in the presence of noise.
Our experiments in Section \ref{sec:experiments} demonstrate that our proposed
estimation algorithms outperform Matrix Pencil, MUSIC and Cadzow's methods.
Both AST and the discretized Lasso algorithms obtain lower MSE
compared to previous approaches, and the discretized algorithm is much
faster on large problems.

\section{Abstract Denoising with Atomic Norms}
\label{sec:abstract-denoising}

The foundation of our technique consists of extending recent work
on \emph{atomic norms} in linear inverse problems in \cite{crpw}. In this work,
the authors describe how to reconstruct models that can be expressed as sparse
linear combinations of \emph{atoms} from some basic set $\A.$ The set $\A$ can
be very general and not assumed to be finite. For example, if the signal is
known to be a low rank matrix, $\A$ could be the set of all unit norm rank-$1$ matrices.

We show how to use an atomic norm penalty to denoise a signal known to be a
sparse nonnegative combination of atoms from a set $\A$. We compute the mean-squared-error for
the estimate we thus obtain and propose an efficient computational method.

\begin{definition}[Atomic Norm]

The atomic norm $\vnorm{\cdot}_\A$ of $\A$ is the Minkowski functional (or the
gauge function) associated with $\conv(\A)$ (the convex hull of $\A$):
\begin{equation}
	\label{defatnorm} \vnorm{x}_\A = \inf\left\{t > 0 ~\middle|~ x \in t \conv(\A) \right\}. 
\end{equation}
\end{definition}

The gauge function is a norm if $\conv(\A)$ is compact, centrally symmetric,
and contains a ball of radius $\epsilon$ around the origin for some
$\epsilon>0$. When $\A$ is the set of unit norm
$1$-sparse elements in $\C^n$, the atomic norm $\vnorm{\cdot}_\A$ is the
$\ell_1$ norm~\cite{candes06}. Similarly, when $\A$ is the set of unit norm
rank-$1$ matrices, the atomic norm is the nuclear norm~\cite{Recht10}.
In~\cite{crpw}, the authors showed that minimizing the atomic norm subject to
equality constraints provided exact solutions of a variety of linear inverse
problems with nearly optimal bounds on the number of measurements required.

To set up the atomic norm denoising problem, suppose we observe a signal $y =
x^\star + w$ and that we know \emph{a priori} that $x^\star$ can be written as
a linear combination of a few atoms from $\A$. One way to estimate $x^\star$
from these observations would be to search over all short linear combinations
from $\A$ to select the one which minimizes $\vnorm{y- x}_2$. However,
this could be formidable: even if the set of atoms is a finite collection of
vectors, this problem is the NP-hard SPARSEST VECTOR
problem~\cite{Natarajan95}.

On the other hand, the problem~\eqref{AST} is convex, and reduces to many
familiar denoising strategies for particular $\A$. The mapping from $y$ to the
optimal solution of \eqref{AST} is called the proximal operator of the atomic
norm applied to $y$, and can be thought of as a soft thresholded version of
$y$. Indeed, when $\A$ is the set of $1$-sparse atoms, the atomic norm is the
$\ell_1$-norm, and the proximal operator corresponds to
\emph{soft-thresholding} $y$ by element-wise shrinking towards
zero~\cite{donoho1995noising}. Similarly, when $\A$ is the set of rank-$1$
matrices, the atomic norm is the nuclear norm and the proximal operator shrinks
the singular values of the input matrix towards zero.

We now establish some universal properties about the problem~\eqref{AST}.
First, we collect a simple consequence of the optimality conditions in a lemma:

\begin{lemma}[Optimality Conditions]
\label{lem:optimality-conditions}
$\hat{x}$ is the solution of \eqref{AST} if and only if\\
\emph{(i)}  $\vnorm{y - \hat{x}}_\A^* \leq \tau$,
\emph{(ii)} $\vabs{y - \hat{x}}{\hat{x}} = \tau \vnorm{\hat{x}}_\A.$
\end{lemma}

The dual atomic norm is given by
\begin{equation}
  \label{defdualatnorm}
  \vnorm{z}_{\A}^* = \sup_{\vnorm{x}_\A \leq 1}{\vabs{x}{z}},
\end{equation}
which implies
\begin{equation}
  \label{holder}
  \vabs{x}{z} \leq \vnorm{x}_\A \vnorm{z}_\A^*.
\end{equation}
The supremum in \eqref{defdualatnorm} is achievable, namely, for any $x$ there is a
$z$ that achieves equality. Since $\A$ contains all extremal points of
$\{x: \|x\|_\A \leq 1\}$, we are guaranteed that the optimal solution will actually lie in the set
$\A$:
\begin{equation}
\vnorm{z}_{\A}^* = \sup_{a \in \A}{\vabs{a}{z}}.\label{bonsall}
\end{equation}

The dual norm will play a critical role throughout, as our asymptotic error
rates will be in terms of the dual atomic norm of noise processes. The dual
atomic norm also appears in the dual problem of \eqref{AST} \begin{lemma}[Dual
Problem] \label{lem:dual-problem} The dual problem of~\eqref{AST} is 
\begin{equation*} \label{eq:dual-ast} \begin{split} &\maximize_z
\frac{1}{2}\left(\vnorm{y}_2^2 - \vnorm{y - z}_2^2\right)\\ &\text{subject to }
\vnorm{z}_\A^* \leq \tau. \end{split} \end{equation*} 
The dual problem admits a
unique solution $\hat{z}$ due to strong concavity of the objective function. The primal solution $\hat{x}$ and the dual solution
$\hat{z}$ are specified by the optimality conditions and there is no duality
gap:\\
\emph{(i)} $y = \hat{x} + \hat{z},$ 
\emph{(ii)} $\vnorm{\hat{z}}_\A^* \leq \tau,$
\emph{(iii)} $\vabs{\hat{z}}{\hat{x}} = \tau
\vnorm{\hat{x}}_\A$.  \
\end{lemma}

The proofs of Lemma \ref{lem:optimality-conditions} and Lemma
\ref{lem:dual-problem} are provided in Appendix
\ref{proof:optimality-conditions}. A straightforward corollary of this Lemma is
a certificate of the support of the solution to \eqref{AST}.

\begin{corollary}[Dual Certificate of Support]
\label{cor:dual-cert-support}

Suppose for some $S \subset \A$,  $\hat{z}$ is a solution to the dual problem \eq{dual-ast} satisfying
\begin{enumerate}
\item $\vabs{\hat{z}}{a} = \tau$ whenever $a \in S,$
\item $|\vabs{\hat{z}}{a}| < \tau$ if $a \not\in S.$
\end{enumerate}
Then, any solution $\hat{x}$ of \eqref{AST} admits a decomposition $\hat{x} =
\sum_{a \in S}{c_a a}$ with $\vnorm{\hat{x}}_\A = \sum_{a \in S}{c_a}.$

\end{corollary}

Thus the dual solution $\hat{z}$ provides a way to determine a
decomposition of $\hat{x}$ into a set of elementary atoms that achieves the
atomic norm of $\hat{x}$. In fact, one could evaluate the inner product
$\left<\hat{z}, a\right>$ and identify the atoms where the absolute value of the
inner product is $\tau$. When the SNR is high, we expect that the
decomposition identified in this manner should be close to the original
decomposition of $x^\star$ under certain assumptions.

We are now ready to state a proposition which gives an upper bound on the MSE
with the optimal choice of the regularization parameter.

\begin{proposition}
\label{prop:main-result}

If the regularization parameter $\tau > \vnorm{w}_\A^*$, the optimal solution
$\hat{x}$ of \eqref{AST} has the MSE
\begin{equation}\label{eq:slow-mse}
\frac{1}{n}\vnorm{\hat{x} - x^\star}_2^2 \leq \frac{1}{n}\left(\tau \vnorm{x^\star}_\A - \vabs{x^\star}{w}\right) \leq \frac{2\tau}{n}\vnorm{x^\star}_\A.
\end{equation} 
\begin{proof}
\begin{align}
	&\vnorm{\hat{x} - x^\star}_2^2 = \vabs{\hat{x}-x^\star}{w-(y -\hat{x})}         \label{expand}\\
	& = \vabs{x^\star}{y-\hat{x}} - \vabs{x^\star}{w} +\vabs{\hat{x}}{w} -\vabs{\hat{x}}{y-\hat{x}}\nonumber\\
	& \leq \tau\vnorm{x^\star}_\A - \vabs{x^\star}{w}+ (\vnorm{w}_\A^* -  \tau)\vnorm{\hat{x}}_\A\label{p1}\\
	& \leq (\tau + \vnorm{w}_\A^*)\vnorm{x^\star}_\A + (\vnorm{w}_\A^* - \tau)\vnorm{\hat{x}}_\A\label{p2}
\end{align}
where for  \eqref{p1} we have used Lemma \ref{lem:optimality-conditions} and \eqref{holder}. 
The theorem now follows from \eqref{p1} and \eqref{p2} since $\tau
> \vnorm{w}_\A^*.$ The value of the regularization parameter $\tau$ to ensure
the MSE is upper bounded thus, is $\vnorm{w}_\A^*.$
\end{proof}
\end{proposition}

\noindent\emph{Example: Sparse Model Selection} We can specialize our stability
guarantee to Lasso~\cite{tibshirani96} and recover known results. Let $\Phi \in
\R^{n \times p}$ be a matrix with unit norm columns, and suppose we observe
$y = x^\star + w$, where $w$ is additive noise, and $x^\star = \Phi
c^\star$ is an unknown $k$ sparse combination of columns of $\Phi$. In this
case, the atomic set is the collection of columns of $\Phi$ and $-\Phi$, and
the atomic norm is $\vnorm{x}_{\A} = \min \left\{\vnorm{c}_1: x = \Phi c\right\}$.
Therefore, the proposed optimization
problem \eqref{AST} coincides with the Lasso estimator~\cite{tibshirani96}.
This method is also known as Basis Pursuit Denoising~\cite{chen01}. If we
assume that $w$ is a gaussian vector with variance $\sigma^2$ for its entries,
the expected dual atomic norm of the noise term, $\vnorm{w}_\A^* =
\vnorm{\Phi^*w}_\infty$ is simply the expected maximum of $p$ gaussian random
variables. Using the well known result on the maximum of gaussian random
variables~\cite{lr76}, we have $\E\vnorm{w}_\A^* \leq \sigma \sqrt{2 \log(p)}$.
If $\hat{x}$ is the denoised signal, we have from Theorem
\ref{cor:expected-mse} that if $\tau = \E\vnorm{w}_\A^* = \sigma \sqrt{2
\log(p)}$, \[ \frac{1}{n}\E\vnorm{\hat{x} - x^\star}_2^2 \leq \sigma
\frac{\sqrt{2\log(p)} }{n} \vnorm{c^\star}_1, \] which is the stability result
for Lasso reported in \cite{greenshtein04} assuming no conditions on $\Phi$.

\subsection{Accelerated Convergence Rates}\label{sec:convergence-rate}
In this section, we provide conditions under which a faster convergence rate
can be obtained for AST.
\begin{proposition}[Fast Rates]
Suppose the set of atoms $\A$ is centrosymmetric and $\vnorm{w}_\A^*$
concentrates about its expectation so that $P( {\vnorm{w}_\A^*} \geq
\E{\vnorm{w}_\A^*}+t) < \delta(t)$. For $\gamma \in [0, 1]$, define the cone
\begin{align*}
C_\gamma(x^\star,\A) &= \cone(\{z:\vnorm{x^\star+ z}_\A \leq \vnorm{x^\star}_\A + \gamma\vnorm{z}_\A\}).
\end{align*}
Suppose 
\begin{equation}
\label{eq:compatibility}
\phi_\gamma(x^\star,\A) := \inf \left\{ \frac{\vnorm{z}_{2}}{\vnorm{z}_\A}  : z \in C_\gamma(x^\star,\A) \right\} 
\end{equation}
is strictly positive for some $\gamma > {\E\vnorm{w}_\A^*}/{\tau}$. Then
\begin{equation}
\label{eq:fast_rate_phi}
\vnorm{\hat{x}-x^\star}_2^2  \leq \frac{(1+\gamma)^2 \tau^2}{\gamma^2 \phi_\gamma(x^\star,\A)^2}
\end{equation}
with probability at least $1-\delta(\gamma\tau - \E\vnorm{w}_\A^*)$.
\end{proposition}

Having the ratio of norms bounded below is a generalization of the Weak
Compatibility criterion used to quantify when fast rates are achievable for the
Lasso~\cite{degeer}. One difference is that we define the corresponding
cone $C_\gamma$ where $\phi_\gamma$ must be controlled in parallel with
the~\emph{tangent cones} studied in~\cite{crpw}. There, the authors showed that
the mean width of the cone $C_0(x^\star,\A)$ determined the number of random
linear measurements required to recover $x^\star$ using atomic norm
minimization. In our case, $\gamma$ is greater than zero, and represents a
``widening'' of the tangent cone. When $\gamma=1$, the cone is all of $\R^n$ or $\C^n$
(via the triangle inequality), hence $\tau$ must be larger than the
expectation to enable our proposition to hold.

\begin{proof}
Since $\hat{x}$ is optimal, we have,
\begin{equation*}
\tfrac{1}{2}\vnorm{y - \hat{x}}_2^2 + \tau \vnorm{\hat{x}}_\A \leq \tfrac{1}{2}\vnorm{y - x^\star}_2^2 + \tau \vnorm{x^\star}_\A
\end{equation*}
Rearranging and using~\eqref{holder} gives
\begin{equation*}
\tau\vnorm{\hat{x}}_\A \leq \tau\vnorm{x^\star}_\A + \vnorm{w}_\A^* \vnorm{\hat{x} - x^\star}_\A.
\end{equation*} 
Since $\vnorm{w}_\A^*$ concentrates about its expectation, with
probability at least $1-\delta(\gamma\tau - \E\vnorm{w}_\A^*)$, we have $\vnorm{w}_\A^* \leq \gamma \tau$ and hence $\hat{x} - x^\star \in
C_\gamma$.
Using \eqref{expand}, if $\tau > \vnorm{w}_\A^*$, 
\begin{align*}
\vnorm{\hat{x}-x^\star}_2^2 & \leq (\tau + \vnorm{w}_\A^*)\vnorm{\hat{x}-x^\star}_\A \leq \frac{(1+\gamma) \tau}{\gamma\phi_\gamma(x^\star,\A)}\vnorm{\hat{x}-x^\star}_2
\end{align*}
So, with probability at least $1-\delta(\gamma\tau - \E\vnorm{w}_\A^*)$:
\begin{equation*}\belowdisplayskip=-10pt
\vnorm{\hat{x}-x^\star}_2^2  \leq \frac{(1+\gamma)^2 \tau^2}{\gamma^2 \phi_\gamma(x^\star,\A)^2}
\end{equation*}
\end{proof}

The main difference between~\eq{fast_rate_phi} and~\eq{slow-mse} is that the
MSE is controlled by $\tau^2$ rather than $\tau \| x^*\|_{\A}$.
As we will now see~\eq{fast_rate_phi} provides minimax optimal rates for the
examples of sparse vectors and low-rank matrices.

\emph{Example: Sparse Vectors in Noise}
Let $\A$ be the set of signed canonical basis
vectors in $\R^n$. In this case, $\conv(\A)$ is the unit cross polytope and the
atomic norm $\vnorm{\cdot}_\A$, coincides with the $\ell_1$ norm, and the dual
atomic norm is the $\ell_\infty$ norm. Suppose $x^\star \in \R^n$ and $T :=
\supp(x^\star)$ has cardinality $k$. Consider the problem of estimating
$x^\star$ from $y = x^\star + w$ where $w \sim \mathcal{N}(0,\sigma^2 I_n).$
 
We show in the appendix that in this case $\phi_\gamma(x^\star,\A)
>\frac{(1-\gamma)}{2\sqrt{k}}$. We also have $\tau_0 = \E \vnorm{w}_\infty \geq
\sigma\sqrt{2\log(n)}.$ Pick $\tau > \gamma^{-1} \tau_0$ for some $\gamma < 1.$
Then, using our lower bound for $\phi_\gamma$ in \eq{fast_rate_phi}, we get a
rate of
\begin{align}\label{eq:sparse-fast-rate}
\frac{1}{n}\vnorm{\hat{x}-x^\star}_2^2 = O\left(\frac{\sigma^2 k\log(n)}{n}\right)
\end{align}
for the AST estimate with high probability. This bound coincides with
the minimax optimal rate derived by Donoho and Johnstone~\cite{Donoho94}. Note
that if we had used~\eq{slow-mse} instead, our MSE would have
instead been $O\left(\sqrt{\sigma^2 k\log n}\|x^\star\|_2/n\,
\right)$, which depends on the norm of the input signal $x^\star$.

\emph{Example: Low Rank Matrix in Noise}
Let $\A$ be the manifold of unit norm rank-$1$ matrices in $\C^{n\times n}$. In
this case, the atomic norm $\vnorm{\cdot}_\A$, coincides with the nuclear norm
$\vnorm{\cdot}_*$, and the corresponding dual atomic norm is the spectral norm
of the matrix. Suppose $X^\star \in \C^{n\times n}$ has rank $r$, so it can be
constructed as a combination of $r$ atoms, and we are interested in estimating
$X^\star$ from $Y = X^\star + W$ where $W$ has independent
$\mathcal{N}(0,\sigma^2)$ entries.

We prove in the appendix that $\phi_\gamma(X^\star,\A) \geq
\frac{1-\gamma}{2\sqrt{2r}}$. To obtain an estimate for $\tau$, we note that
the spectral norm of the noise matrix satisfies $\|W\|\leq
2\sqrt{n}$ with high probability~\cite{Davidson01}. Substituting these
estimates for $\tau$ and $\phi_\gamma$ in \eq{fast_rate_phi}, we get
the minimax optimal MSE
\begin{align*}
\frac{1}{n^2}\vnorm{X-\hat{X}}_F^2 = O\left( \frac{\sigma^2 r}{n} \right).
\end{align*}
\subsection{Expected MSE for Approximated Atomic Norms}
\label{proof:expected-mse-approx}

We close this section by noting that it may sometimes be easier to solve
\eqref{AST} on a different set $\widetilde{\A}$ (say, an $\epsilon$-net of
$\A$ instead of $\A$. If for some $M>0,$
\[
M^{-1}\vnorm{x}_{\widetilde{\A}} \leq \vnorm{x}_{\A} \leq \vnorm{x}_{\widetilde{\A}}
\] 
holds for every $x$, then Theorem $\ref{cor:expected-mse}$ still applies with a constant factor $M$. We will need the following lemma.

\begin{lemma}
${\vnorm{z}_{\A}^* \leq M\vnorm{z}_{\widetilde{\A}}^*}$ for every $z$ iff
${M^{-1}\vnorm{x}_{\widetilde{\A}} \leq \vnorm{x}_{\A}}$ for every $z$.
\end{lemma}
\begin{proof}\belowdisplayskip=-12pt
We will show the forward implication -- the converse will follow since the dual
of the dual norm is again the primal norm. By \eqref{holder}, for any $x$, there exists a $z$ with
$\vnorm{z}_{\widetilde{\A}}^* \leq 1$ and ${\vabs{x}{z} =
\vnorm{x}_{\widetilde{\A}}}$. So,
\begin{align*}
M^{-1}\vnorm{x}_{\widetilde{\A}} &= M^{-1}\vabs{x}{z} &&\\
&\leq M^{-1}\vnorm{z}_{\A}^* \vnorm{x}_{\A} &&\text{by \eqref{holder}}\\
&\leq \vnorm{x}_{\A} &&\text{by assumption.}
\end{align*}
\end{proof}
Now, we can state the sufficient condition for the following proposition in
terms of either the primal or the dual norm:
\begin{proposition}\label{prop:grid-approx-mse}
Suppose 
\begin{equation}
  \label{dni}
  \vnorm{z}_{\widetilde{\A}}^* \leq \vnorm{z}_{\A}^* \leq M\vnorm{z}_{\widetilde{\A}}^* \text{ for every } z,
\end{equation}
or equivalently 
\begin{equation}
  \label{pni}
  M^{-1}\vnorm{x}_{\widetilde{\A}} \leq \vnorm{x}_{\A} \leq \vnorm{x}_{\widetilde{\A}} \text{ for every } x,
\end{equation}
then under the same conditions as in Theorem \ref{cor:expected-mse},
\begin{equation*}
    \frac{1}{n} \E \vnorm{\tilde{x} - x^\star}_2^2 \leq \frac{M \tau}{n}\vnorm{x^\star}_\A
\end{equation*}
where $\tilde{x}$ is the optimal solution for \eqref{AST} with $\A$ replaced by $\widetilde{\A}.$
\end{proposition}
\begin{proof}\belowdisplayskip=-12pt
By assumption, $\E\left(\vnorm{w}_{\A}^*\right) \leq \tau$. Now, \eqref{dni}
implies $\E\left(\vnorm{w}_{\widetilde{\A}}^*\right) \leq \tau.$ Applying
Theorem \ref{cor:expected-mse}, and using \eqref{pni}, we get
\begin{equation*}
\frac{1}{n} \E \vnorm{\tilde{x} - x^\star}_2^2 \leq \frac{\tau}{n}\vnorm{x^\star}_{\widetilde{\A}} \leq \frac{M \tau}{n}\vnorm{x^\star}_{\A}.
\end{equation*}
\end{proof}

\section{Application to Line Spectral Estimation}
\label{sec:denoise-trig-moments} 

Let us now return to the line spectral estimation problem, where we denoise a linear combination of complex sinusoids. The atomic set in this
case consists of samples of individual sinusoids, $a_{f,\phi} \in \C^n$,
given by
\begin{equation}
\label{eq:trig-atoms} a_{f,\phi} = e^{i2\pi \phi}\begin{bmatrix}1 ~ e^{i2\pi f} ~
\cdots ~ e^{i2\pi(n-1)f} \end{bmatrix}^T\,.
\end{equation}
The  infinite set $\A = \{ a_{f,\phi}: f \in
[0,1], \phi \in [0,1] \}$ forms an appropriate collection of atoms for
$x^\star$, since $x^\star$ in \eq{tru-moment} can be written as a sparse
nonnegative combination of atoms in $\A.$ In fact, $x^\star = \sum_{l = 1}^k
c_l^\star a_{f_l^\star,0} = \sum_{l = 1}^k |c_l^\star| a_{f_l^\star,\phi_l},$
where $c_l^\star = |c_l^\star|e^{i2\pi\phi_l}.$

The corresponding dual norm takes an intuitive form:
\begin{align}
	\|v\|_{\A}^* &= \sup_{{f,\phi}} \langle v, a_{f,\phi} \rangle=\sup_{f\in [0,1]}  \sup_{\phi \in [0,1]}  e^{i 2\pi \phi} \sum_{l=0}^{n-1} v_l e^{-2\pi i l f} =\sup_{ |z|\leq 1 }  \left| \sum_{l=0}^{n-1} v_l z^l\right|.\label{eq:dual-norm-poly}
\end{align}
In other words, $\|v\|_{\A}^*$ is the maximum absolute value attained on the
unit circle by the polynomial $\zeta \mapsto \sum_{l=0}^{n-1} v_l \zeta^l$.  Thus, in what follows, we will frequently refer to the \emph{dual polynomial} as the polynomial whose coefficients are given by the dual optimal solution of the AST problem.

\subsection{SDP for Atomic Soft Thresholding}
\label{sec:sdp-ast}

In this section, we present a semidefinite characterization of the atomic norm
associated with the line spectral atomic set $\mathcal{A} = \{a_{f,\phi} | f
\in [0, 1], \phi \in [0, 1]\}$. This characterization allows us to rewrite
 {\eqref{AST}} as an equivalent semidefinite programming problem.

Recall from \eqref{eq:dual-norm-poly} that the dual atomic norm of a vector $v
\in \mathbb{C}^n$ is the maximum absolute value of a complex trigonometric
polynomial $V(f) = \sum_{l=0}^{n-1} v_l e^{-2\pi i l f}$. As a
consequence, a constraint on the dual atomic norm is equivalent to
a bound on the magnitude of $V(f)$:
\begin{align*}
\|v\|_\A^* \leq \tau \Leftrightarrow |V(f)|^2 \leq \tau^2, \forall f \in [0, 1].
\end{align*}
The function $q(f) = \tau^2-|V(f)|^2$ is a trigonometric polynomial (that is, a
polynomial in the variables $z$ and $z^*$ with $|z|=1$). A necessary and
sufficient condition for $q(f)$ to be nonnegative is that it can be written as
a sum of squares of trigonometric polynomials~\cite{Megretski03}. 
Testing if $q$ is a sum of squares can be achieved
via semidefinite programming. To state the associated semidefinite program,
define the map $T:\mathbb{C}^n \rightarrow \mathbb{C}^{n\times n}$ which
creates a Hermitian Toeplitz matrix out of its input, that is
\[
T(x)= \left[
\begin{array}{ccccc} x_1 & x_2 & \ldots & x_n\\ 
x^*_2 & x_1  & \ldots & x_{n-1}\\
 \vdots & \vdots & \ddots & \vdots\\
 x^*_n & x^*_{n-1}  & \ldots & x_1
 \end{array}\right]
\]
Let $T^*$ denote the adjoint of the map $T$. Then we have the following
succinct characterization

\begin{lemma}\cite[Theorem 4.24]{brl2007}\label{lm:brl} For any given causal trigonometric polynomial $V(f) = \sum_{l=0}^{n-1} v_l
e^{-2\pi i l f}$, $|V(f)| \leq \tau $ if and only if there exists complex
Hermitian matrix $Q$ such that
\begin{align*}
T^*(Q) = \tau^2 {e}_1~~\mbox{and}~~
\begin{bmatrix}
  Q & v \\
  v^* & 1
 \end{bmatrix} \succeq 0.
\end{align*}
Here, ${e}_1$ is the first canonical basis vector with a one at the first
component and zeros elsewhere and $v^*$ denotes the Hermitian adjoint
(conjugate transpose) of $v$.
\end{lemma}

Using Lemma \ref{lm:brl}, we rewrite the atomic norm $\|x\|_\A =
\sup_{\|v\|_\A^*\leq 1} \left<x, v\right> $ as the following semidefinite
program:
\begin{equation}\label{eq:sdpprimal}
 \begin{array}{ll}
\operatorname*{maximize}_{v,\ Q} & \left<x, v\right>\\
\text{subject  to}  & T^*(Q) = {e}_1 \\
& \begin{bmatrix}
  Q & v \\
  v^* & 1
 \end{bmatrix} \succeq 0.
\end{array}
\end{equation}
The dual problem of \eq{sdpprimal} (after a trivial rescaling) is then equal to
the atomic norm of $x$:
\begin{align*}
\begin{array}{lll} \|x\|_\A =&  \min_{t, u} & \tfrac{1}{2} (t + u_1)  \\
&\operatorname{subject\ to}
& \begin{bmatrix}
  T(u) & x \\
  x^* & t
 \end{bmatrix} \succeq 0.\end{array}
\end{align*}
Therefore, the atomic denoising problem \eqref{AST} for the set of trigonometric atoms is equivalent to
\begin{equation}\label{eq:sdpdenoising}
\begin{array}{ll}
\operatorname*{minimize}_{t, u, x} & \frac{1}{2} \|x - y\|_2^2 + \frac{\tau}{2}(t + u_1) \\
\operatorname{subject\ to}
& \begin{bmatrix}
  T(u) &  x \\
 x^* & t
 \end{bmatrix} \succeq 0.\end{array} 
\end{equation}

The semidefinite program \eq{sdpdenoising} can be solved by off-the-shelf
solvers such as SeDuMi~\cite{sedumi} and SDPT3~\cite{SDPT3}. However, these
solvers tend to be slow for  large problems. For the interested reader, we provide a reasonably efficient algorithm based upon the Alternating Direction Method of Multipliers
(ADMM) \cite{admm2011} in Appendix~\label{sec:admm}

\subsection{Choosing the regularization parameter}\label{subsec:parameter}
The choice of the regularization parameter is dictated by the noise model and
we show the optimal choice for white gaussian noise samples in our analysis. As
noted in Theorem~\ref{cor:expected-mse}, the optimal choice of the
regularization parameter depends on the dual norm of the noise.  A simple
lower bound on the expected dual norm occurs when we consider the maximum
value of $n$ uniformly spaced points in the unit circle. Using the result of
\cite{lr76}, the lower bound whenever $n \geq 5$ is
\[
\sigma\sqrt{n\log(n) - \tfrac{n}{2} \log(4\pi\log(n))}\,.
\]

Using a theorem of Bernstein and standard results on the extreme value
statistics of Gaussian distribution, we can also obtain a non-asymptotic upper
bound on the expected dual norm of noise for $n > 3$:
\[\sigma\left(1  + \frac{1}{\log(n)}\right)\sqrt{n \log(n) + n\log( 4\pi \log(n))}\nonumber
\]
(See Appendix \ref{proof:dual-norm-bounds} for a derivation of both the lower
and upper bound). If we set the regularization parameter $\tau$ equal to an
upper bound on the expected dual atomic norm, i.e.,
\begin{equation}
\label{eq:tau}
\tau = \sigma\left(1  +  \frac{1}{\log(n)}\right)\sqrt{n \log(n) + n\log(4\pi\log(n))}.
\end{equation}
an application of Theorem \ref{cor:expected-mse} yields the asymptotic result
in Theorem \ref{thm:expmsels}.

\subsection{Determining the frequencies}
\label{sec:frequency-localize}
\begin{figure}[t]
\centering
\includegraphics[width=2.7in]{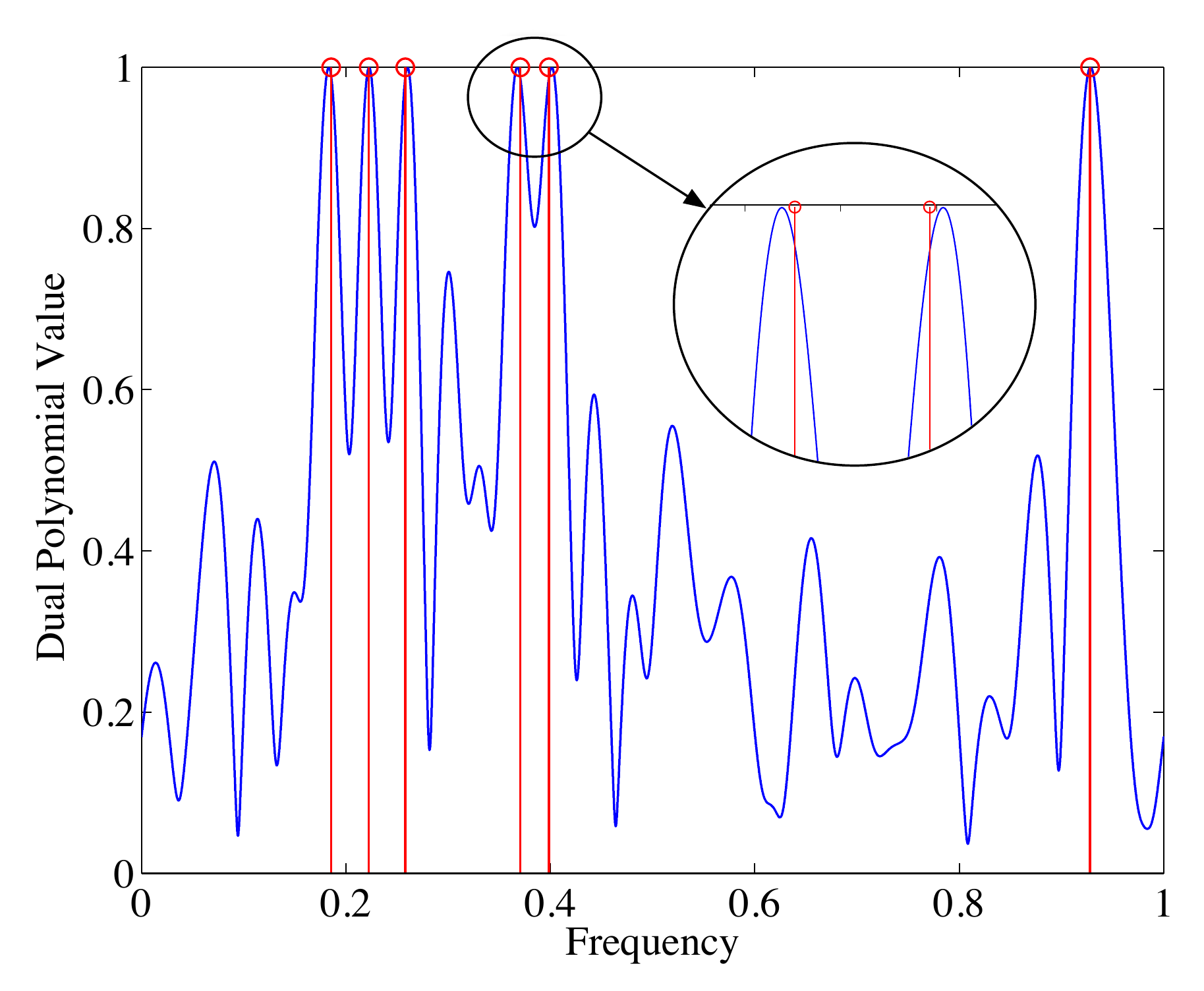}
\caption{ {\bf Frequency Localization using Dual Polynomial}: The
actual location of the frequencies in the line spectral signal $x^\star \in
\C^{64}$ is shown in red. The blue curve is the dual polynomial
obtained by solving \eq{dual-ast} with $y = x^\star + w$ where $w$ is noise of SNR 10 dB.}
\label{fig:dual_poly_localize}
\end{figure}

As shown in Corollary~\ref{cor:dual-cert-support}, the dual solution can be
used to identify the frequencies of the primal solution. For line spectra, a frequency
$f \in [0,1]$ is in the support of the solution $\hat{x}$ of \eqref{AST} if and
only if
\[
	 |\langle \hat{z}, a_{f,\phi} \rangle| =\left|\ \sum_{l=0}^{n-1} \hat{z}_l e^{-i 2\pi l f} \right| = \tau
\]
That is, $f$ is in the support of $\hat{x}$ if and only if it is a point of maximum modulus for the dual polynomial. Thus, the support
may be determined by finding frequencies $f$ where the dual polynomial attains magnitude $\tau$. 

Figure~\ref{fig:dual_poly_localize} shows the dual
polynomial for $\eqref{AST}$ with $n = 64$ samples and $k = 6$
 randomly chosen frequencies. The regularization
parameter $\tau$ is chosen as described in Section \ref{subsec:parameter}.

A recent result by Candes and Fernandez-Granda~\cite{CandesGranda} establishes
that in the noiseless case, the frequencies localized by the dual polynomial
are exact provided the minimum separation between the frequencies is at least
$4/n$ where $n$ is the number of samples in the line spectral signal. Under similar separation condition, numerical simulations suggest that \eqref{AST} achieves approximate frequency location in the noisy case. 

\subsection{Discretization and Lasso}
\label{sec:comp-method}
When the number of samples is larger than a few hundred, the running time of
our ADMM method is dominated by the cost of computing eigenvalues and is
usually expensive~\cite{12BhaskarArxiv}. For very large problems, we now propose using Lasso as an
alternative to the semidefinite program~\eq{sdpdenoising}. To proceed, pick a
uniform grid of $N$ frequencies and form $\A_N = \left\{ a_{m/N,\phi}
~\middle|~ 0 \leq m \leq N-1 \right\} \subset \A $ and solve \eqref{AST} on
this grid. i.e., we solve the problem
\begin{equation}
	\label{epsprimal} \text{minimize }\frac{1}{2} \vnorm{x - y}_2^2 + \tau \vnorm{x}_{\A_N}. 
\end{equation}

To see why this is to our advantage, define $\Phi$ be the $n \times N$ Fourier
matrix with $m$th column $a_{m/N,0}$. Then for any $x \in \C^n$ we have $\vnorm{x}_{\A_N} = \min\left\{ \vnorm{c}_1: x = \Phi c \right\}$.
So, we solve
\begin{equation}
	\label{sparsa} \text{minimize }\frac{1}{2} \vnorm{\Phi c- y}_2^2 + \tau \vnorm{c}_1. 
\end{equation}
for the optimal point $\hat{c}$ and set $\hat{x}_N = \Phi \hat{c}$ or the first
$n$ terms of the $N$ term discrete Fourier transform (DFT) of $\hat{c}$.
Furthermore, $\Phi^* z$ is simply the $N$ term inverse DFT of $z \in \C^n$.
This observation coupled with Fast Fourier Transform (FFT) algorithm for
efficiently computing DFTs gives a fast method to solve \eqref{epsprimal},
using standard compressed sensing software for $\ell_2-\ell_1$ minimization,
for example, SparSA~\cite{wright09}.

Because of the relatively simple structure of the atomic set, the optimal
solution $\hat{x}$ for \eqref{epsprimal} can be made arbitrarily close to
\eqref{eq:sdpdenoising} by picking $N$ a constant factor larger than $n$. In
fact, we show that the atomic norms on $\A$ and $\A_N$ are equivalent (See
Appendix \ref{proof:dual-norm-approximation}):
\begin{equation}
 \left(1-\frac{2\pi n}{N}\right) \vnorm{x}_{\A_N} \leq  \vnorm{x}_\A \leq \vnorm{x}_{\A_N}, \forall x \in \C^n
\end{equation}
Using Proposition
\ref{prop:grid-approx-mse} and \eq{tau}, we conclude
{\small
\begin{align*}
\frac{1}{n} \E \vnorm{\hat{x}_N - x^\star}_2^2 
&\leq
\frac{\sigma\left(\frac{\log(n)+1}{\log(n)}\right)\vnorm{x^\star}_\A
\sqrt{
    n \log(n) + 
    n\log(4\pi\log(n))
}}{n\left(1 - \frac{2\pi n}{N}\right)}=
O\left(
\sigma \sqrt{\frac{\log(n)}{n}} \frac{ \vnorm{x^\star}_\A}{\left(1 - \frac{2\pi n}{N}\right)}
\right)
\end{align*}
}

Due to the efficiency of the FFT, the discretized approach
has a much lower algorithmic complexity than either Cadzow's alternating
projections method or the ADMM method described in Appendix~\ref{sec:admm},
which each require computing an eigenvalue decomposition at
each iteration. Indeed, fast solvers for~\eqref{sparsa} converge to an
$\epsilon$ optimal solution in no more than $1/\sqrt{\epsilon}$ iterations.
Each iteration requires a multiplication by $\Phi$ and a simple ``shrinkage''
step. Multiplication by $\Phi$ or $\Phi^*$ requires $O(N\log N)$ time and the
shrinkage operation can be performed in time $O(N)$.

As we discuss below, this fast form of basis pursuit has been proposed by
several authors. However, analyzing this method with tools from compressed
sensing has proven daunting because the matrix $\Phi$ is nowhere near a
restricted isometry. Indeed, as $N$ tends to infinity, the columns become more
and more coherent. However, common sense says that a larger grid should give
better performance, for both denoising and frequency localization! Indeed, by appealing to the atomic norm framework, we are
able to show exactly this point: the larger one makes $N$, the closer one
approximates the desired atomic norm soft thresholding problem. Moreover, we do
not have to choose $N$ to be too large in order to achieve nearly the same
performance as the AST.

\section{Related Work}
\label{sec:prony-method}

The classical methods of line spectral estimation, often called linear
prediction methods, are built upon the seminal interpolation method of
Prony~\cite{prony1795}. In the noiseless case, with as little as $n=2k$
measurements, Prony's technique can identify the frequencies exactly, no matter
how close the frequencies are. However, Prony's technique is known to be
sensitive to noise due to instability of polynomial rooting~\cite{kahn92}.
Following Prony, several methods have been employed to robustify polynomial
rooting method including the Matrix Pencil algorithm~\cite{hua02}, which
recasts the polynomial rooting as a generalized eigenvalue problem and cleverly
uses extra observations to guard against noise. The MUSIC~\cite{music} and
ESPRIT~\cite{esprit} algorithms exploit the low rank structure of the
autocorrelation matrix.

Cadzow~\cite{cadzow02} proposed a heuristic that improves over MUSIC by 
exploiting the Toeplitz structure of the matric of moments by alternately
projecting between the linear space of Toeplitz matrices and the space of rank
$k$ matrices where $k$ is the desired model order.
Cadzow's technique is very similar~\cite{ssa_special_issue} to a popular
technique in time series literature~\cite{tsbook1,tsbook2} called Singular
Spectrum Analysis~\cite{ssa}, which uses autocorrelation matrix instead
of the matrix of moments for projection. Both these techniques may be viewed as
instances of structured low rank approximation~\cite{chu2003structured} which
exploit additional structure beyond low rank structure used in subspace based
methods such as MUSIC and ESPRIT. Cadzow's method has been identified as a  fruitful preprocessing step for linear prediction methods~\cite{blu08}. A  survey of classical linear prediction methods can be found in~\cite{blu08,StoicaMoses} and an extensive list of references is given in~\cite{stoica93}.

Most, if not all of the linear prediction methods need to estimate the model
order by employing some heuristic and the performance of the algorithm is
sensitive to the model order. In contrast, our algorithms AST and the Lasso
based method, only need a rough estimate of the noise variance. In our experiments, we provide
the true model order to Matrix Pencil, MUSIC and Cadzow methods, while we use
the estimate of noise variance for AST and Lasso methods, and still compare favorably to the classical line spectral methods.

In contrast to linear prediction methods, a number of authors
~\cite{chen98spectrum,malioutov05,bourguignon2007irregular}
have suggested using compressive sensing and viewing the frequency estimation
as a sparse approximation problem. For instance,~\cite{malioutov05} notes that
the Lasso based method has better empirical localization performance than the
popular MUSIC algorithm. However, the theoretical analysis of this phenomenon
is complicated because of the need to replace the continuous frequency space by
an oversampled frequency grid. Compressive sensing based results (see, for
instance,~\cite{duartescs}) need to carefully control the incoherence of their
linear maps to apply off-the-shelf tools from compressed sensing. It is
important to note that the performance of our algorithm improves as the grid
size increases. But this seems to contradict conventional wisdom in compressed
sensing because our design matrix $\Phi$ becomes more and more coherent, and
limits how fine we can grid for the theoretical guarantees to hold.

We circumvent the problems in the conventional compresssive sensing analysis by
directly working in the continuous parameter space and hence step away from
such notions as coherence, focussing on the geometry of the atomic set as the
critical feature. By showing that the continuous approach is the limiting
case of  the Lasso based methods using the convergence of the corresponding
atomic norms, we justify denoising line spectral signals using Lasso on a
large grid. Since the original submission of this manuscript,
Cand\`es and Fernandez-Granda~\cite{CandesGranda} showed that our SDP
formulation exactly recovers the correct frequencies in the noiseless case.

\section{Experiments}
\label{sec:experiments}

We compared the MSE performance of AST, the discretized Lasso approximation, 
the Matrix Pencil, MUSIC and Cadzow's method. For our experiments, we
generated $k$ normalized frequencies $f_1^\star, \ldots, f_k^\star$ uniformly
randomly chosen from $\left[0,1\right]$ such that every pair of frequencies are
separated by at least $1/2n$. The
signal $x^\star \in \C^n$ is generated according to \eq{tru-moment} with
$k$ random amplitudes independently chosen from $\chi^2(1)$ distribution (squared Gaussian).
All of our sinusoids were then assigned a random phase (equivalent to 
multiplying $c_k^\star$ by a random unit norm complex number). Then, the observation
 $y$ is produced by adding complex white gaussian noise $w$ such that the
input signal to noise ratio (SNR) is $-10,-5,0,5,10,15$ or $20$ dB. We compare
the average MSE of the various algorithms in 10 trials for various values of number
of observations $(n = 64,128,256)$, and number of frequencies $ (k = n/4,n/8,n/16)$.

AST needs an estimate of the noise variance $\sigma^2$ to pick the regularization
parameter according to \eq{tau}. In many situations, this variance is not known to us
\emph{a priori}. However, we can construct a reasonable estimate for $\sigma$
when the phases are uniformly random. It is known that the autocorrelation
matrix of a line spectral signal (see, for example Chapter 4 in
\cite{StoicaMoses}) can be written as a sum of a low rank matrix and $\sigma^2
I$ if we assume that the phases are uniformly random. Since the empirical
autocorrelation matrix concentrates around the true expectation, we can
estimate the noise variance by  averaging a few
smallest eigenvalues of the empirical autocorrelation matrix. In the following experiments, we form the empirical
autocorrelation matrix using the MATLAB routine {\tt corrmtx} using a
prediction order $m=n/3$ and averaging the lower $25\%$ of the eigenvalues. We
used this estimate in equation~\eq{tau} to determine the regularization
parameter for both our AST and Lasso experiments.

First, we implemented AST using the ADMM method described in detail
in the Appendix. We used the stopping criteria described
in~\cite{admm2011} and set $\rho=2$ for all experiments. We use the dual
solution $\hat{z}$ to determine the support of the optimal solution $\hat{x}$
using the procedure described in Section~\ref{sec:frequency-localize}. Once the
frequencies $\hat{f}_l$ are extracted, we ran the least squares problem
$\mbox{minimize}_\alpha \|U \alpha - y\|^2$ where $U_{jl} = \exp(i 2\pi j
\hat{f}_l)$ to obtain a \emph{debiased} solution. After computing the optimal
solution $\alpha_{\mathrm{opt}}$, we returned the prediction $\hat{x} =
U\alpha_{\mathrm{opt}}$.

We implemented Lasso, obtaining an estimate $\hat{x}$ of $x^\star$ from $y$ by
solving the optimization problem \eqref{epsprimal} with debiasing. We use the
algorithm described in Section \ref{sec:comp-method} with grid of $N=2^{m}$
points where $m=10,11,12,13,14$ and $15$. Because of the basis mismatch effect,
the optimal $c_{\mathrm{opt}}$ has significantly more non-zero components than
the true number of frequencies. However, we observe that the frequencies
corresponding to the non-zero components of $c_{\mathrm{opt}}$ cluster
around the true ones. We therefore extract one frequency from each cluster of
non-zero values by identifying the grid point with the maximum absolute
$c_{\mathrm{opt}}$ value and zero everything else in that cluster. We then ran
a debiasing step which solves the least squares problem $\mbox{minimize}_\beta
\|\Phi_S \beta- y\|^2$ where $\Phi_S$ is the submatrix of $\Phi$ whose columns
correspond to frequencies identified from $c_{\mathrm{opt}}$. We return the
estimate $\hat{x} = \Phi_S \beta_{\mathrm{opt}}$. We used the freely
downloadable implementation of SpaRSA to solve the Lasso problem. We used a
stopping parameter of $10^{-4}$, but otherwise used the default parameters.

We implemented Cadzow's method as described by the
pseudocode in~\cite{blu08}, the Matrix Pencil as described in~\cite{hua02} and
MUSIC~\cite{music} using the MATLAB routine {\tt rootmusic}. All these
algorithms need an estimate of the number of sinusoids. Rather
than implementing a heuristic to estimate $k$, \emph{we fed the true $k$ to our
solvers}. This provides a huge advantage to these algorithms. Neither AST or the
Lasso based algorithm are provided the true value of $k$, and the noise 
variance $\sigma^2$ required in the regularization parameter is estimated from $y$.

\begin{figure*}
  \begin{tabular}{ccc}
\includegraphics[width=0.3\textwidth]{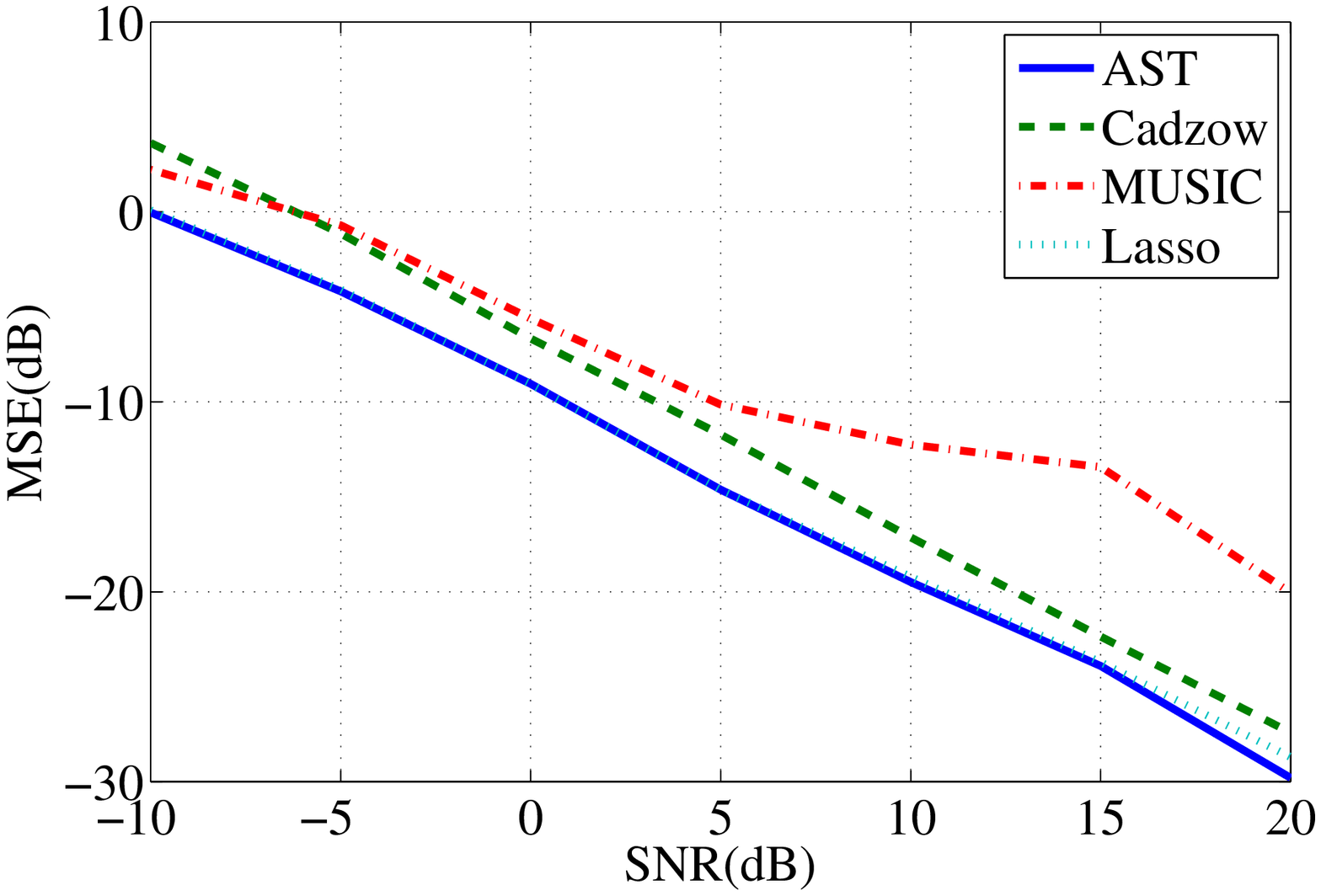} &
\includegraphics[width=0.3\textwidth]{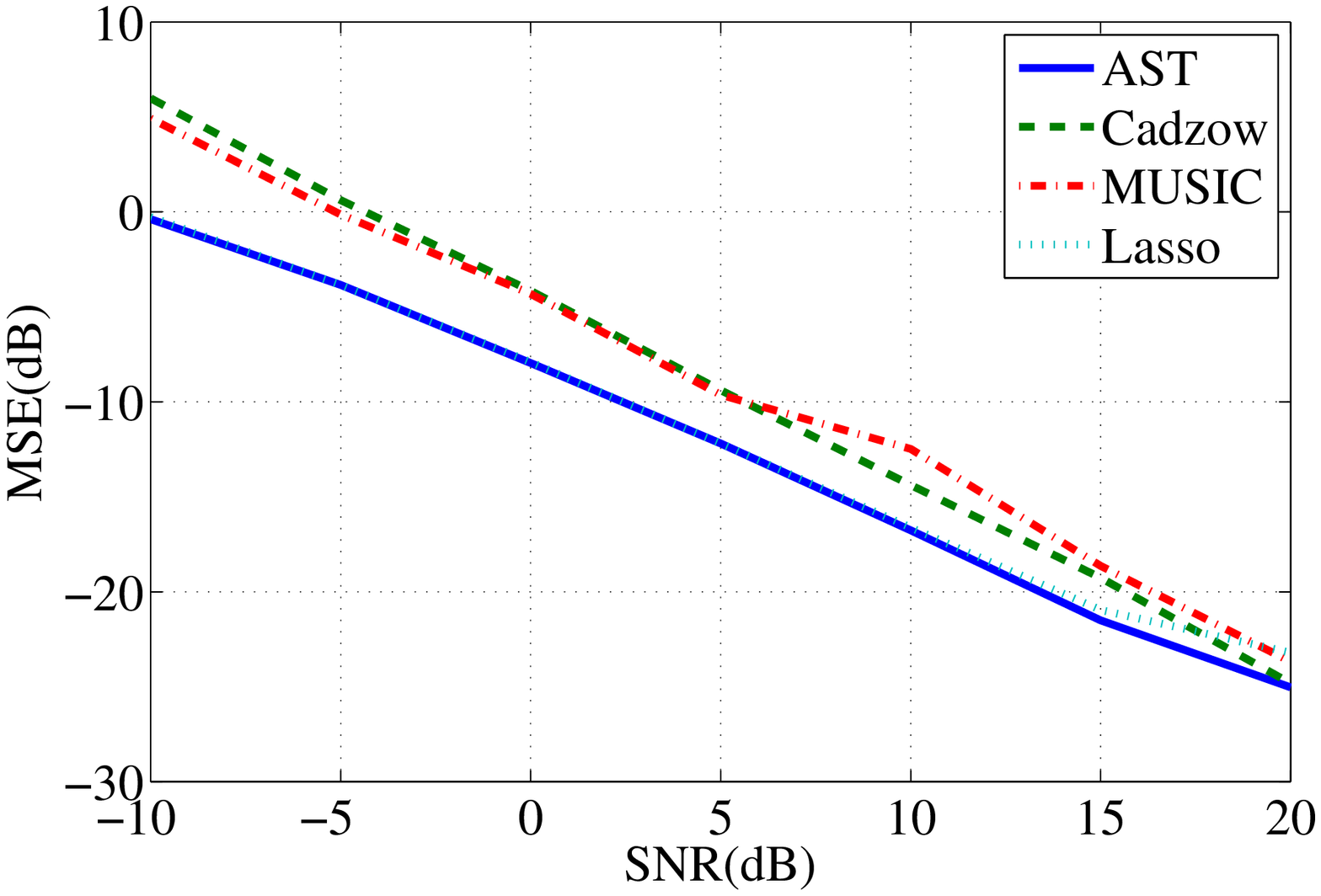} &
\includegraphics[width=0.3\textwidth]{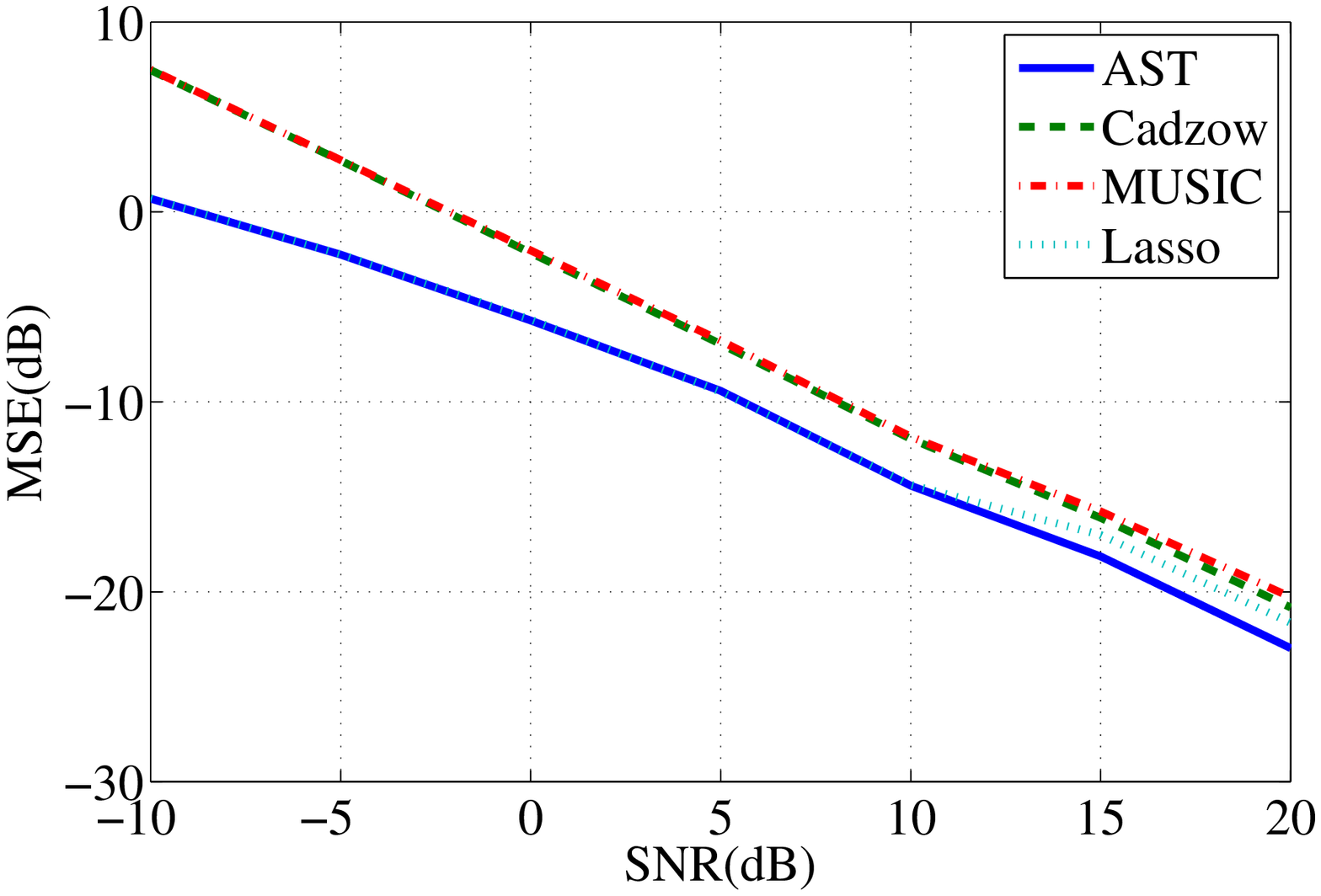}
\end{tabular}
\caption{ {\bfseries MSE vs SNR plots:} This graph compares MSE vs SNR for a subset of experiments with $n=128$ samples.  From left to right, the plots are for combinations of $8$, $16$, and $32$ sinusoids with amplitudes and frequencies sampled at random.}
\label{fig:mse-snr}
\end{figure*}

In Figure~\ref{fig:mse-snr}, we show MSE vs SNR plots for a subset of
experiments when $n=128$ time samples are taken to take a closer look at the
differences. It can be seen from these plots that the performance difference
between classical algorithms such as MUSIC and Cadzow with respect to the
convex optimization based AST and Lasso is most pronounced at lower sparsity
levels. When the noise dominates the signal (SNR $\leq 0$ dB), all the
algorithms are comparable. However, AST and Lasso outperform the other
algorithms in almost every regime.

We note that the denoising performance of Lasso improves with increased
grid size as shown in the MSE vs SNR plot in Figure~\ref{fig:lasso-compare}(a).  The 
figure  shows that the  performance improvement for larger grid sizes is greater at high SNRs. 
This is because when the noise is small, the discretization error is more dominant and
finer gridding helps to reduce this error. Figures~\ref{fig:lasso-compare}(a) and (b)
also indicate that the benefits of increasing discretization levels are diminishing with the
grid sizes, at a higher rate in the low SNR regime, suggesting a tradeoff among
grid size, accuracy,  and computational complexity.

Finally, in Figure \ref{fig:lasso-compare}(b), we provide numerical evidence supporting the
assertion that frequency localization improves with increasing grid size.  Lasso identifies more
frequencies than the true ones due to basis mismatch. However, these
frequencies cluster around the true ones, and more importantly, finer
discretization improves clustering, suggesting over-discretization coupled with
clustering and peak detection as a means for frequency localization for Lasso.
This observation does not contradict  the results of
\cite{cpsc} where
the authors look at the full Fourier basis ($N=n$) and the noise-free
case.  This is the situation where discretization effect is most prominent.  We
instead look at the scenario where $N \gg n$.

\begin{figure}
\centering
  \begin{tabular}{cc}
 \includegraphics[trim=10mm 0mm 5mm 3mm,clip,width=0.3\linewidth]{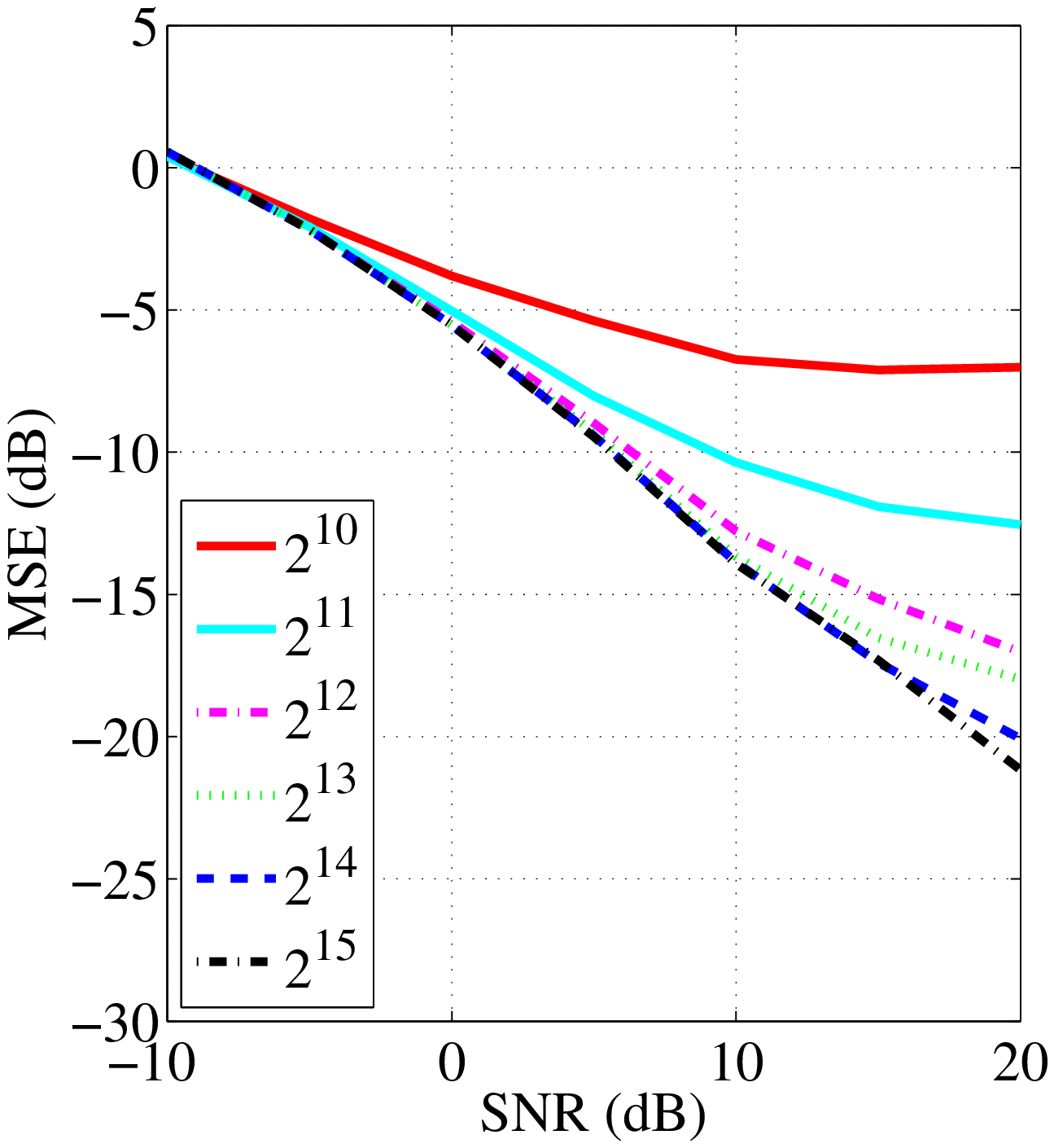} &
  \includegraphics[trim=15mm 0mm 20mm 3mm,clip,width=0.34\linewidth]{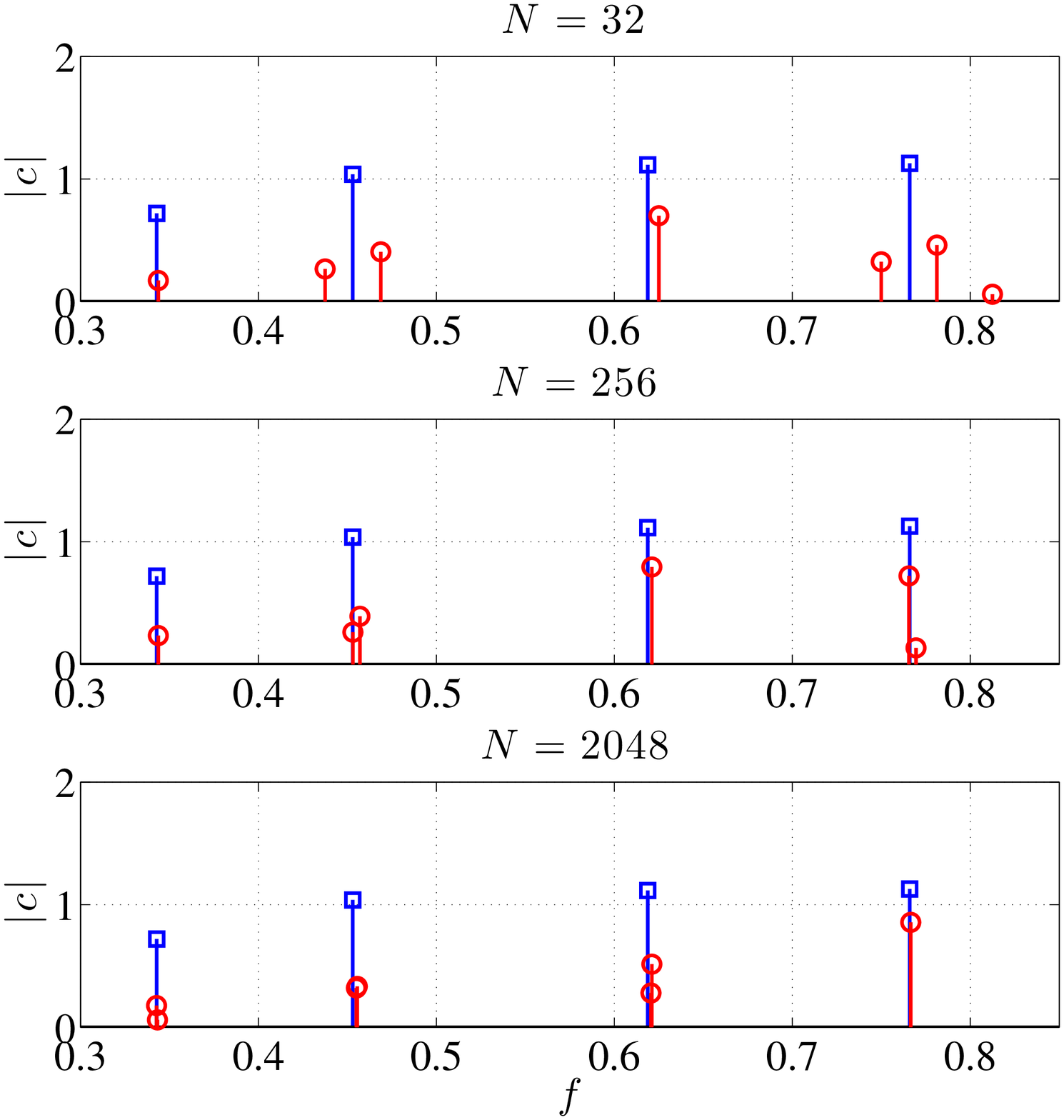} \\
(a)  & (b)
\end{tabular}
\caption{ 
(a) Plot of MSE vs SNR for Lasso at different grid sizes for a subset of experiments with $n=128, k = 16$.
(b) Lasso Frequency localization with $n=32, k = 4,$ SNR = $10$ dB. Blue represents the true frequencies, while red are given by Lasso. For better visualization, we threshold the Lasso solution by $10^{-6}$.
}
\label{fig:lasso-compare}
\end{figure}
We use \emph{performance profiles} to summarize the behavior of the various algorithms
across all of the parameter settings.  Performance profiles provide a good visual indicator
of the relative performance of many algorithms under a variety of experimental
conditions\cite{dolanmore02}. Let $\mathcal{P}$ be the set of experiments and
let $\mathrm{MSE}_s(p)$ be the MSE of experiment $p \in
\mathcal{P}$ using the algorithm $s$. Then the ordinate $P_s(\beta)$ of the
graph at $\beta$ specifies the fraction of experiments where the ratio of the
MSE of the algorithm $s$ to the minimum MSE across all algorithms for the given
experiment is less than $\beta$, i.e.,
\begin{equation*}
P_s(\beta) = \frac{\mathop{\#}\left\{p \in \mathcal{P} ~:~ \mathrm{MSE}_s(p) \leq \beta \min_s \mathrm{MSE}_s(p)\right\}}{\mathop{\#}(\mathcal{P})}
\end{equation*}

From the performance profile in Figure~\ref{fig:pp}(a), we see that AST is the
best performing algorithm, with Lasso coming in
second. Cadzow does not perform as well as AST, even though it is fed the true
number of sinusoids. When Cadzow is fed an incorrect $k$, even off by $1$, the
performance degrades drastically, and never provides adequate mean-squared
error. Figure~\ref{fig:pp}(b) shows that the denoising performance
improves  with grid size.

\begin{figure}
\centering
\begin{tabular}{cc}
\includegraphics[width=.4\linewidth]{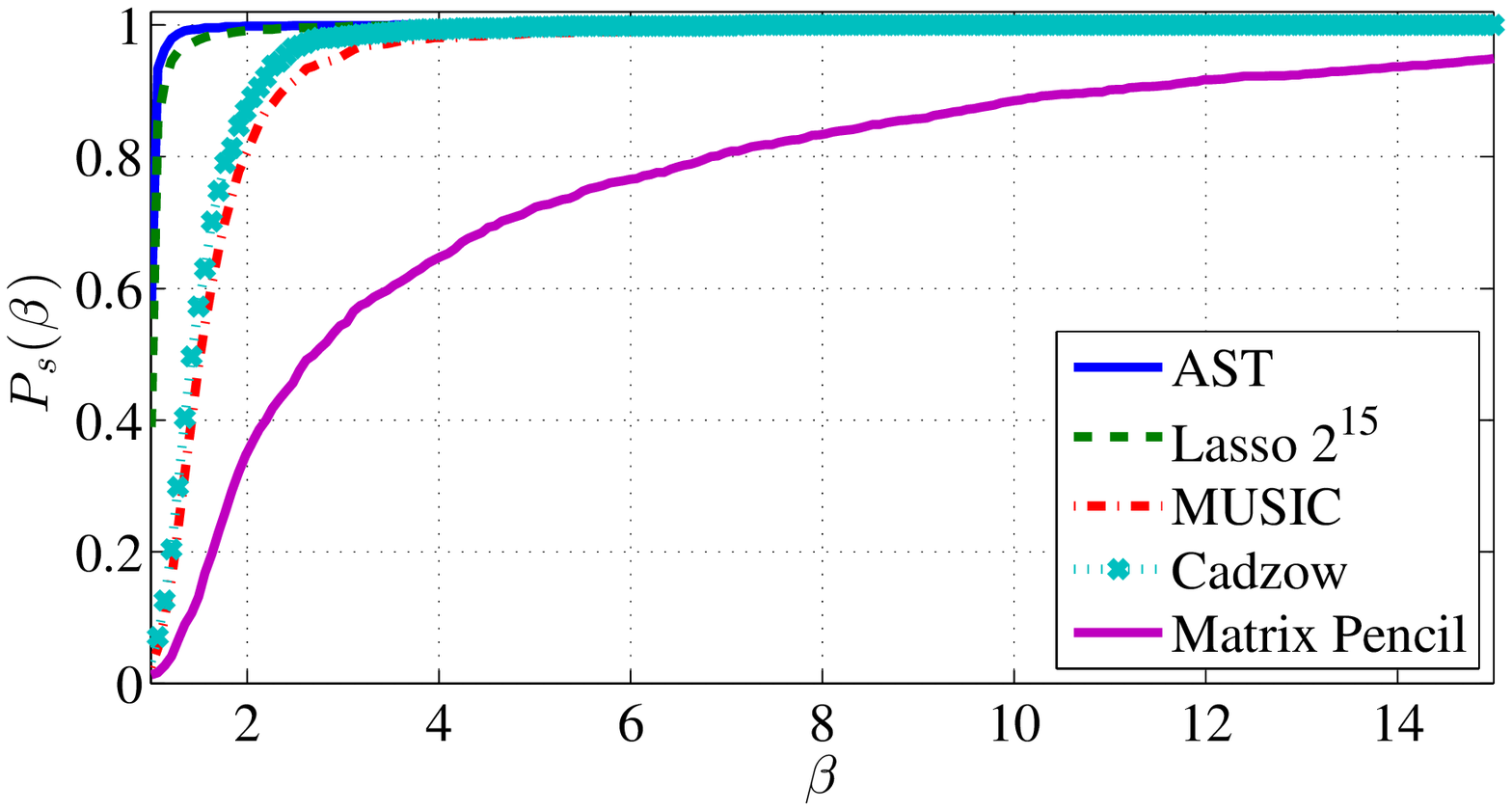} &
\includegraphics[trim=0mm 0mm 2mm 5mm,clip,width=.4\linewidth]{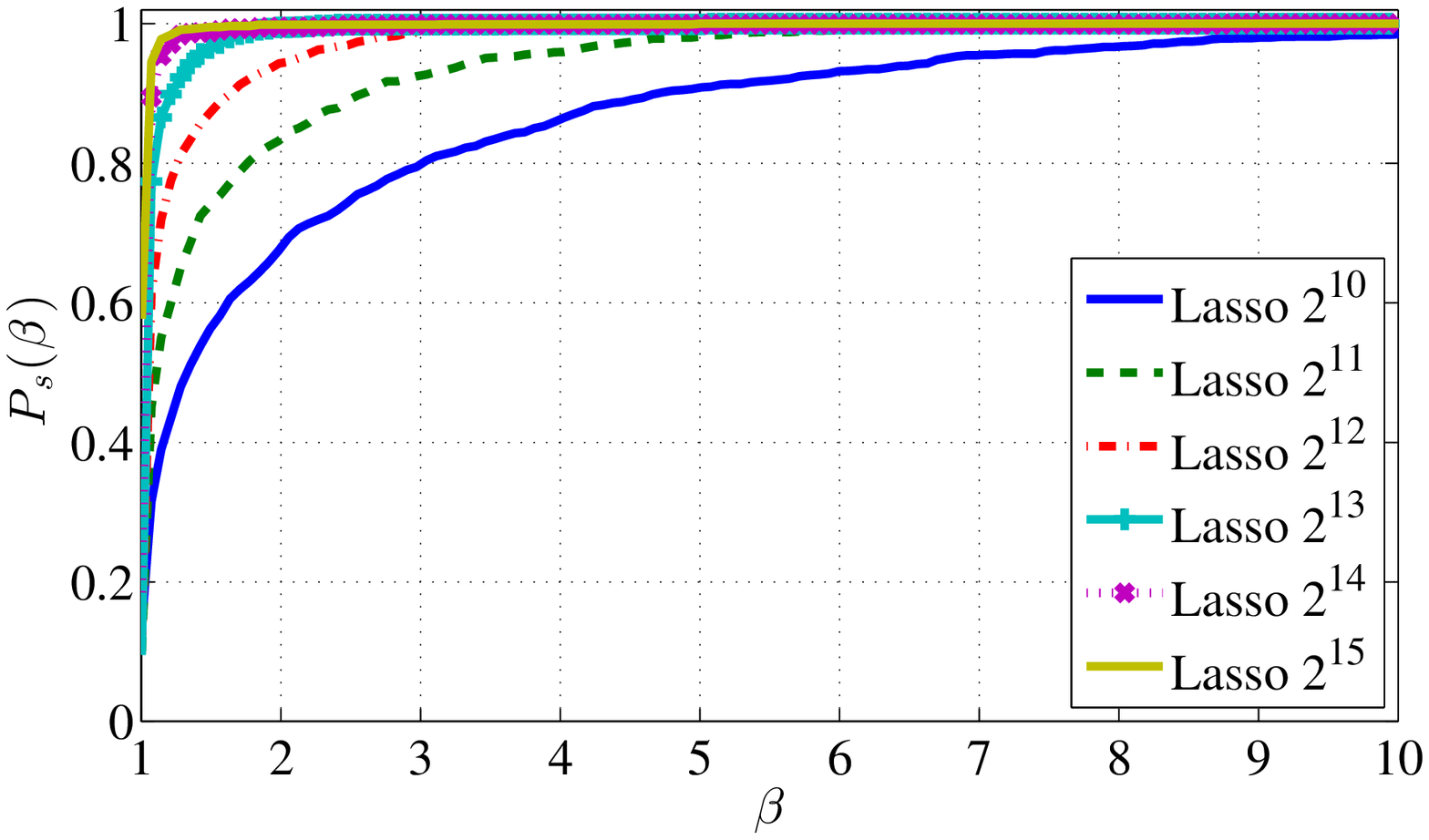}\\
(a) & (b)
\end{tabular}
\caption{(a) Performance Profile  comparing various algorithms and AST. (b) Performance profiles for Lasso with different grid sizes.}
\label{fig:pp}
\end{figure}

\section{Conclusion and Future Work}\label{sec:conclusions}

The Atomic norm formulation of line spectral estimation provides several
advantages over prior approaches. By performing the
analysis in the continuous domain we were able to derive simple closed form
rates using fairly straightforward techniques.We only grid the unit circle at the very end of
our analysis and determine the loss incurred from discretization. This approach
allowed us to circumvent some of the more complicated theoretical arguments
that arise when using concepts from compressed sensing or random matrix theory.

This work provides several interesting possible future directions, both in line
spectral estimation and in signal processing in general. We conclude with a
short outline of some of the possibilities.

\paragraph{Fast Rates} Determining checkable conditions on the cones in
Section~\ref{sec:convergence-rate} for the atomic norm problem is a major open
problem. Our experiments suggest that when the frequencies are spread out, AST
performs much better with a slightly larger regularization parameter.  This observation was also made in the model-based compressed sensing literature~\cite{duartescs}.  Moreover, Cand\`es and Fernandez Granda also needed a spread assumption to prove their theories.

This evidence together suggests the \emph{fast rate} developed in Section~\ref{sec:convergence-rate} may be active for signals with well separated frequenies. Determining concrete conditions on the signal $x^\star$ that ensure this fast rate require techniques for estimating the parameter $\phi$ in~\eq{compatibility}. Such an
investigation should be accompanied by a determination of the minimax rates for
line spectral estimation. Such minimax rates would shed further light on the
rates achievable for line spectral estimation.

\paragraph {Moments Supported Inside the Disk} Our work also naturally extends
to moment problems where the atomic measures are supported on the unit disk in
the complex plane. These problems arise naturally in controls and systems
theory and include model order reduction, system identification, and control
design. Applying the standard program developed in
Section~\ref{sec:abstract-denoising} provides a new look at these classic
operator theory problems in control theory. It would be of significant
importance to develop specialized atomic-norm denoising algorithms for control
theoretic problems. Such an approach could yield novel statistical bounds for
estimation of rational functions and $\mathcal{H}_\infty$-norm approximations.

\paragraph{Other Denoising Models} Our abstract denoising results in
Section~\ref{sec:abstract-denoising} apply to any atomic models and it is worth
investigating their applicability for other models in statistical signal
processing. For instance, it might be possible to pose a scheme for denoising a
signal corrupted by multipath reflections. Here, the atoms might be all time
and frequency shifted versions of some known signal. It remains to be seen what
new insights in statistical signal processing can be gleaned from our unified
approach to denoising.

\section*{Acknowledgements} The authors would like to thank Vivek Goyal,
Parikshit Shah, and Joel Tropp for many helpful conversations and suggestions
on improving this manuscript. This work was supported in part by NSF Award
CCF-1139953 and ONR Award N00014-11-1-0723.

{\small
\bibliographystyle{ieeetr}
\bibliography{trbib}
}

\appendix
\section{Optimality Conditions}
\label{proof:optimality-conditions}
\subsubsection{Proof of Lemma \ref{lem:optimality-conditions}}
\begin{proof}
The function $f(x) = \tfrac{1}{2}\vnorm{y - x}_2^2 + \tau\vnorm{x}_\A$ is minimized at $\hat{x}$, if for all $\alpha \in (0,1)$ and all $x$,
\begin{eqnarray}
&& f(\hat{x} + \alpha(x - \hat{x})) \geq f(\hat{x})\nonumber
\end{eqnarray}
or equivalently,
\begin{eqnarray}
 \alpha^{-1} \tau \left(\vnorm{\hat{x} + \alpha(x - \hat{x})}_\A - \vnorm{\hat{x}}_\A  \right) \geq \vabs{y - \hat{x}}{x - \hat{x}} - \frac{1}{2}\alpha \vnorm{x - \hat{x}}_2^2\label{eq:limit}
\end{eqnarray}
Since $\vnorm{\cdot}_\A$ is convex, we have
\[
\vnorm{x}_\A - \vnorm{\hat{x}}_\A \geq \alpha^{-1}  \left(\vnorm{\hat{x} + \alpha(x - \hat{x})}_\A - \vnorm{\hat{x}}_\A \right),
\]
for all $x$ and for all $\alpha \in (0,1)$. Thus, by letting $\alpha \to 0$ in \eq{limit}, we note that $\hat{x}$ minimizes $f(x)$ only if, for all $x$,
\begin{equation}
\tau \left( \vnorm{x}_\A - \vnorm{\hat{x}}_\A \right) \geq \vabs{y - \hat{x}}{x - \hat{x}}.\label{eq:subgradient}
\end{equation}
However if \eq{subgradient} holds, then, for all $x$
\begin{align*}
\frac{1}{2}\| y - x \|_2^2 + \tau \vnorm{x}_\A
\geq \frac{1}{2}\| y - \hat{x} + (\hat{x} - x) \|_2^2 + \vabs{y - \hat{x}}{x - \hat{x}} + \tau \vnorm{\hat{x}}_\A
 \end{align*}
implying $f(x) \geq f(\hat{x}).$
Thus, \eq{subgradient} is necessary and sufficient for $\hat{x}$ to minimize $f(x)$. 
\begin{note}
The condition \eq{subgradient} simply says that $\tau^{-1} \left(y - \hat{x} \right)$ is in the subgradient of $\vnorm{\cdot}_\A$ at $\hat{x}$ or equivalently that $0 \in \partial f(\hat{x})$.
\end{note}

We can rewrite \eq{subgradient} as
\begin{equation}
\label{eq:reformulate-subgradient}
\tau \vnorm{\hat{x}}_\A - \vabs{y - \hat{x}}{\hat{x}} \leq \inf_x \left\{ \tau \vnorm{x}_\A - \vabs{y - \hat{x}}{x} \right\}
\end{equation}
But by definition of the dual atomic norm, 
\begin{align}
\label{eq:conjugate-norm}
\sup_x \left\{ \vabs{z}{x} - \vnorm{x}_\A \right\} &=  I_{\left\{ w : \vnorm{w}_\A^* \leq 1\right\}}(z)
 = 
\begin{cases}
0 & \vnorm{z}_\A^* \leq 1\\
\infty & \text{otherwise.}
\end{cases}
\end{align}
where $I_A(\cdot)$ is the convex indicator function. Using this in \eq{reformulate-subgradient}, we find that $\hat{x}$ is a minimizer if and only if $\vnorm{y - \hat{x}}_\A^* \leq \tau$ and $\vabs{y - \hat{x}}{\hat{x}} \geq \tau \vnorm{\hat{x}}_\A$. This proves the theorem.
\end{proof}

\subsubsection{Proof of Lemma \ref{lem:dual-problem}}
\begin{proof}
We can rewrite the primal problem \eqref{AST} as a constrained optimization problem:
\[
\begin{split}
&\minimize_{x,u} \frac{1}{2}\vnorm{y - x}_2^2 + \vnorm{u}_\A\\
&\text{subject to }  u = x.
\end{split}
\]
Now, we can introduce the Lagrangian function
\begin{equation*}
L(x,u,z) = \frac{1}{2}\vnorm{y - x}_2^2 + \vnorm{u}_\A + \vabs{z}{x - u}.
\end{equation*} 
so that the dual function is given by
\begin{align*}
g(z) = \inf_{x,u} L(x,u,z)
& = \inf_{x}\left( \frac{1}{2}\vnorm{y - x}_2^2 + \vabs{z}{x} \right) 
+ \inf_{u}\left( \tau\vnorm{u}_\A - \vabs{z}{u} \right)\\
& = \frac{1}{2}\left(\vnorm{y}_2^2 - \vnorm{y - z}_2^2\right) - I_{\left\{ w : \vnorm{w}_\A^* \leq \tau\right\}}(z).
\end{align*}
where the first infimum follows by completing the squares and the second infimum follows from \eq{conjugate-norm}. Thus the dual problem  of maximizing $g(z)$ can be written as in \eq{dual-ast}.

The solution to the dual problem is the unique projection $\hat{z}$ of $y$ on to the closed convex set $C = \{ z : \vnorm{z}_\A^* \leq \tau \}$. By projection theorem for closed convex sets, $\hat{z}$ is a projection of $y$ onto $C$ if and only if $\hat{z} \in C$ and $\vabs{z - \hat{z}}{y - \hat{z}} \leq 0$ for all  $z \in C$, or equivalently if $\vabs{\hat{z}}{y - \hat{z}} \geq \sup_z\,\vabs{z}{y - \hat{z}} = \tau \vnorm{y - \hat{z}}_\A.$
These conditions are satisfied for $\hat{z} = y - \hat{x}$ where $\hat{x}$ minimizes $f(x)$ by Lemma \ref{lem:optimality-conditions}. Now the proof follows by the substitution $\hat{z} = y - \hat{x}$ in the previous lemma. The absence of duality gap can be obtained by noting that the primal objective function at $\hat{x},$
\[
f(\hat{x}) = \frac{1}{2}\vnorm{y - \hat{x}}_2^2 + \vabs{\hat{z}}{\hat{x}}
= \frac{1}{2}\vnorm{\hat{z}}_2^2 + \vabs{\hat{z}}{\hat{x}} = g(\hat{z}).
\]
\end{proof}

\section{Fast Rate Calculations}

We first prove the following

\begin{proposition}
Let $\A = \{\pm e_1, \ldots, \pm e_n\},$  be the set of signed canonical unit vectors in $\R^n$. Suppose $x^\star \in \R^n$ has $k$ nonzeros.  Then $\phi_\gamma(x^\star,\A) \geq \frac{(1-\gamma)}{2\sqrt{k}}$.
\end{proposition}

\begin{proof}
Let $z \in C_\gamma(x^\star,\A).$ For some $\alpha>0$ we have,
\[
\vnorm{x^\star + \alpha z}_1 \leq \vnorm{x^\star}_1 + \gamma \vnorm{\alpha z}_1
\]
In the above inequality, set $z = z_T + z_{T^c}$ where $z_T$ are the components on the support of $T$ and $z_{T^c} $ are the components on the complement of $T$. Since $x^\star + z_T$ and $z_{T^c}$ have disjoint supports, we have,
\begin{align*}
\vnorm{x^\star + \alpha z_T}_1 + \alpha \vnorm{z_{T^c}}_1 \leq \vnorm{x^\star}_1 + \gamma \vnorm{\alpha z_T}_1 + \gamma \vnorm{\alpha z_{T^c}}_1\,.
\end{align*}
This inequality  implies
\begin{align*}
 \vnorm{z_{T^c}}_1 \leq \frac{1+\gamma}{1-\gamma} \vnorm{z_T}_1\,
\end{align*}
that is, $z$ satisfies the null space property with a constant of $\tfrac{1+\gamma}{1-\gamma}.$ Thus,
\[
\vnorm{z}_1 \leq \frac{2}{1-\gamma}\vnorm{z_T}_1 \leq \frac{2\sqrt{k}}{1-\gamma}\vnorm{z}_2
\]
This gives the desired lower bound.
\end{proof}

Now we can turn to the case of low rank matrices.
\begin{proposition}
Let $\A$ be the manifold of unit norm rank-$1$ matrices in $\C^{n\times n}$. Suppose $X^\star \in \C^{n\times n}$ has rank $r$.  Then $\phi_\gamma(X^\star,\A) \geq \frac{1-\gamma}{2\sqrt{2r}}$. 
\end{proposition}

\begin{proof}
Let $U \Sigma V^H$ be a singular value decomposition of $X^\star$ with $U \in \C^{n \times r}$, $V \in \C^{n \times r}$ and $\Sigma \in \C^{r \times r}$. Define the subspaces 
\begin{align*}
T &= \{U X + Y V^H ~:~ X, Y \in \C^{n \times r} \}\\
T_0 &= \{U M V^H ~:~ M \in \C^{r \times r} \}
\end{align*}
and let $\mathcal{P}_{T_0}$, $\mathcal{P}_{T}$, and $\mathcal{P}_{T^\perp}$ be projection operators that respectively map onto the subspaces $T_0$, $T$, and the orthogonal complement of $T$. Now, if $Z \in C_\gamma(X^\star, \A)$, then for some $\alpha > 0$,  we have
\begin{align}
\label{eq:cone_matrix}
\vnorm{X^\star + \alpha Z}_* \leq 
\vnorm{X^\star}_* + \gamma \alpha \vnorm{Z}_*
\leq \vnorm{X^\star}_* + \gamma \alpha \vnorm{\mathcal{P}_T (Z)}_* +  \gamma \alpha \vnorm{\mathcal{P}_{T^\perp} (Z)}_*.
\end{align}
Now note that we have
\[
\vnorm{X^\star + \alpha Z}_* \geq \vnorm{X^\star + \alpha\mathcal{P}_{T_0} (Z)}_* + \alpha\vnorm{\mathcal{P}_{T^\perp} (Z)}_*
\]
Substituting this in \eq{cone_matrix}, we have,
\begin{align*}
\vnorm{X^\star + \alpha\mathcal{P}_{T_0} (Z)}_* + \alpha\vnorm{\mathcal{P}_{T^\perp} (Z)}_* 
\leq \vnorm{X^\star}_* + \gamma \alpha \vnorm{\mathcal{P}_T (Z)}_* +  \gamma \alpha \vnorm{\mathcal{P}_{T^\perp} (Z)}_*.
\end{align*}
Since $\vnorm{\mathcal{P}_{T_0} (Z)}_* \leq \vnorm{\mathcal{P}_T (Z)}_*$, we have 
\[
\vnorm{\mathcal{P}_{T^\perp} (Z)}_* \leq \frac{1+\gamma}{1-\gamma}\vnorm{\mathcal{P}_T(Z)}_*.
\]
Putting these computations together gives the estimate
\begin{align*}
\vnorm{Z}_* &\leq \vnorm{\mathcal{P}_T(Z)}_* + \vnorm{\mathcal{P}_{T^\perp}(Z)}_* \leq \frac{2}{1-\gamma}\vnorm{\mathcal{P}_T(Z)}
 \leq \frac{2\sqrt{2r}}{1-\gamma}\vnorm{\mathcal{P}_T(Z)}_F
 \leq \frac{2\sqrt{2r}}{1-\gamma}\vnorm{Z}_F\,.
\end{align*}
That is, we have $\phi_\gamma(X^\star,\A) \geq \frac{1-\gamma}{2\sqrt{2r}}$ as desired.
\end{proof} 

\section{Approximation of the Dual Atomic Norm}
\label{proof:dual-norm-approximation}
This section furnishes the proof that the atomic norms induced by $\A$ and $\A_N$ are equivalent. Note that the dual atomic norm of $w$ is given by
\begin{equation}
  \label{eq:maximum-modulus}
  \vnorm{w}_\A^* = \sqrt{n}\sup_{f \in [0,1]} \left| W_n(e^{i2\pi f}) \right|.
\end{equation}
i.e., the maximum modulus of the polynomial $W_n$ defined by
\begin{equation}
\label{eq:random-poly}
W_n(e^{i2\pi f}) =\frac{1}{\sqrt{n}} \sum_{m=0}^{n-1}{w_m e^{-i 2 \pi m f}}.
\end{equation}
Treating $W_n$ as a function of $f$,  with a slight abuse of  notation, define
\begin{equation*}
\vnorm{W_n}_\infty := \sup_{f \in [0,1]} \left| W_n(e^{i2\pi f}) \right|.
\end{equation*}
We show that we can approximate the maximum modulus by evaluating $W_n$ in a
uniform grid of $N$ points on the unit circle. To show that as $N$ becomes large,
the approximation is close to the true value, we bound the derivative of $W_n$ 
using Bernstein's inequality for polynomials.

\begin{theorem}[Bernstein, See, for example ~\cite{schaeffer41}]
Let $p_n$ be any polynomial of degree $n$ with complex coefficients. Then,
\begin{equation*}
 \sup_{|z|\leq 1} |p'(z)|  \leq n  \sup_{|z|\leq 1} |p(z)|.
\end{equation*}
\end{theorem}
Note that for any $f_1, f_2 \in [0,1],$ we have
\begin{align*}
  \left|W_n(e^{i2\pi f_1})\right| - \left|W_n(e^{i2\pi f_2})\right|  &\leq   \left|e^{i2\pi f_1} - e^{i2\pi f_2}\right| \vnorm{W_n'}_\infty&&\\
  &= 2|\sin(2\pi (f_1-f_2))| \vnorm{W_n'}_\infty\\
  &\leq 4 \pi (f_1-f_2) \vnorm{W_n'}_\infty\label{diff}\\
  & \leq 4 \pi n (f_1-f_2) \vnorm{W_n}_\infty,
\end{align*}
where the last inequality follows by Bernstein's theorem.
Letting $s$ take any of the $N$ values $0,\tfrac{1}{N} \ldots, \tfrac{N-1}{N}$, we see that,
\begin{equation*}
\vnorm{W_n}_\infty \leq \max_{m=0,\ldots,N-1}\left| W_n\left(e^{i 2 \pi m/N}\right) \right| + \frac{2\pi n}{N}\vnorm{W_n}_\infty.
\end{equation*}
Since the maximum on the grid is a lower bound for maximum modulus of $W_n$, we have
\begin{align}
\max_{m=0,\ldots,N-1} \left| W_n\left(e^{i 2 \pi m/N}\right) \right| \leq \vnorm{W_n}_\infty \\
\leq  \left(1-\frac{2\pi n}{N}\right)^{-1} \max_{m=0,\ldots,N-1} \left| W_n\left(e^{i 2 \pi m/N}\right) \right|\nonumber\\
 \leq  \left(1+\frac{4\pi n}{N}\right) \max_{m=0,\ldots,N-1} \left| W_n\left(e^{i 2 \pi m/N}\right) \right|.
\label{eq:grid-approx}
\end{align}
Thus, for every $w,$
\begin{equation}
\vnorm{w}_{\A_N}^* \leq  \vnorm{w}_\A^* \leq  \left(1-\frac{2\pi n}{N}\right)^{-1} \vnorm{w}_{\A_N}^*
\end{equation}
or equivalently, for every $x,$
\begin{equation}
 \left(1-\frac{2\pi n}{N}\right) \vnorm{x}_{\A_N} \leq  \vnorm{x}_\A \leq \vnorm{x}_{\A_N}
\end{equation}

\section{Dual Atomic Norm Bounds}
\label{proof:dual-norm-bounds}
This section derives non-asymptotic upper and lower bounds on the expected dual norm of gaussian noise vectors, which are asymptotically tight upto $\log\log$ factors. Recall that the dual atomic norm of $w$ is given by $\sqrt{n}\sup_{f \in [0,1]}|W_f|$ where
\begin{equation*}
W_f = \frac{1}{\sqrt{n}}\sum_{m=0}^{n-1}{w_m e^{-i2 \pi m f}}.
\end{equation*}
The covariance function of $W_f$ is
\begin{align*}
\E\left[W_{f_1} W_{f_2}^*\right] &= \frac{1}{n}\sum_{m=0}^{n-1} \exp(2\pi m (f_1 - f_2)) = e^{\pi (n - 1) (f_1-f_2)} \frac{\sin(n \pi (f_1 - f_2) ) }{ n \sin(\pi (f_2 - f_2))}.
\end{align*}
Thus, the $n$ samples $\left\{W_{m/n}\right\}_{m=0}^{n-1}$ are uncorrelated and thus independent because of their joint gaussianity. This gives a simple non-asymptotic lower bound using the known result for maximum value of $n$ independent gaussian random variables~\cite{lr76} whenever $n > 5$:
\begin{align*}
\E\left[\sup_{t \in T} \left| W_t \right| \right] &\geq \E\left[\max_{m = 0, \ldots, n-1} \Re\left( W_{m/n} \right)  \right] = \sqrt{\log(n) - \tfrac{ \log\log(n) + \log(4\pi)}{2}}.
\end{align*}

 We will show that the lower bound is asymptotically tight neglecting $\log\log$ terms. Since the dual norm induced by  $\A_N$ approximates the dual norm induced by $\A$, (See \ref{proof:dual-norm-approximation}), it is sufficient to compute an upper bound for $\vnorm{w}_{\A_N}^*.$ Note that $|W_f|^2$ has a chi-square distribution since $W_f$ is a Gaussian process. We establish a simple lemma about the maximum of chi-square distributed random variables.
\begin{lemma}
\label{lem:max-chi}
Let $x_1,\ldots,x_N$ be complex gaussians with unit variance. Then,
\begin{equation*}
\E\left[\max_{1\leq i\leq N} |x_i|\right] \leq \sqrt{\log(N) + 1}.
\end{equation*}
\begin{proof}
Let $x_1,\ldots,x_N$ be complex Gaussians with unit variance: $\E[ |x_i|^2]=1$.  Note that $2|x_i|^2$ is a chi-squared random variable with two degrees of freedom.   Using Jensen's inequality, also observe that
\begin{align}\label{eq:jensen-bound}
	\E\left[\max_{1\leq i\leq N} |x_i|\right] & \leq  
	\E\left[\max_{1\leq i\leq N} |x_i|^2\right]^{1/2} \leq  
	\tfrac{1}{\sqrt{2}}\E\left[\max_{1\leq i\leq N} 2|x_i|^2\right]^{1/2}.
\end{align}

Now let $z_1,\ldots, z_n$ be chi-squared random variables with $2$ degrees of freedom.  Then we have
\begin{align*}
	\E\left[\max_{1\leq i\leq N} z_i\right] &= \int_0^\infty P\left[ \max_{1\leq i\leq N} z_i \geq t \right] dt\\
	\leq & \delta + \int_\delta^\infty P\left[ \max_{1\leq i\leq N} z_i \geq t \right] dt\\
	\leq & \delta +  N \int_\delta^\infty P\left[  z_1 \geq t\right] dt\\
	= & \delta +  N \int_\delta^\infty   \exp(-t/2) dt\\
	= & \delta +  2N  \exp(-\delta/2)
\end{align*}
Setting $\delta = 2\log(N)$ gives $\E\left[\max_{1\leq i\leq N} z_i\right] \leq 2 \log{N} + 2$.  Plugging this estimate into~\eqref{eq:jensen-bound} gives $\E\left[\max_{1\leq i\leq N} |x_i|\right] \leq  \sqrt{\log{N}+1}$.
\end{proof}
\end{lemma}
Using Lemma \ref{lem:max-chi}, we can compute
\begin{align*}
\vnorm{w}_{\A_N}^* & = \sqrt{n}\max_{m=0,\ldots,N-1} \left| W_n\left(e^{i 2 \pi m/N}\right) \right|  \leq \sigma\sqrt{n\left(\log{N}+1\right)}
\end{align*}
Plugging in $N = 4\pi n \log(n)$ and using \eq{maximum-modulus} and \eq{grid-approx} establishes a tight upper bound.

\section{Alternating Direction Method of Multipliers for AST}\label{sec:admm}

A thorough survey of the ADMM algorithm is given in~\cite{admm2011}. We only
present the details essential to the implementation of atomic norm soft
thresholding. To put our problem in an appropriate form for ADMM,
rewrite~\eq{sdpdenoising} as
\begin{equation*}
\begin{array}{ll}
\operatorname*{minimize}_{t, u, x,Z} & \frac{1}{2} \|x - y\|_2^2 + \frac{\tau}{2}(t + u_1) \\
\operatorname{subject\ to}
& Z=\begin{bmatrix}
  T(u) & x \\
  x^* & t
 \end{bmatrix} \\
& Z\succeq 0.\end{array} 
\end{equation*}
and dualize the equality constraint via an Augmented Lagrangian:
\begin{align*}
\mathcal{L}_\rho (t,u,x,Z, \Lambda)= \frac{1}{2} \|x - y\|_2^2 + \frac{\tau}{2}(t +  u_1) + \\
\left\langle   \Lambda, Z-\begin{bmatrix}
  T(u) & x \\
  x^* & t
 \end{bmatrix} \right\rangle +
  \frac{\rho}{2} \left\| Z-\begin{bmatrix}
  T(u) & x \\
  x^* & t
 \end{bmatrix} \right\|_F^2
\end{align*}

ADMM then consists of the update steps:
\begin{align*}
(t^{l+1},u^{l+1},x^{l+1})  \leftarrow \arg\min_{t,u,x} \mathcal{L}_\rho(t,u,x,Z^l, \Lambda^l) \\
Z^{l+1}  \leftarrow \arg\min_{Z\succeq 0}  \mathcal{L}_\rho(t^{l+1},u^{l+1},x^{l+1}, Z, \Lambda^l ) \\
\Lambda^{l+1}  \leftarrow \Lambda^l + \rho  \left( Z^{l+1}-\begin{bmatrix}
  T(u^{l+1}) &  x^{l+1} \\
  {x^{l+1}}^* & t^{l+1}
 \end{bmatrix} \right).
\end{align*}
The updates with respect to $t$, $x$, and $u$ can be computed in closed form:
\begin{align*}
    t^{l+1} &= Z_{n+1,n+1}^l+\left(\Lambda_{n+1,n+1}^l-\frac{\tau}{2}\right)/\rho\\
	 x^{l+1} &= \frac{1}{2\rho+1}(y  + 2\rho z_1^{l} + 2\lambda_1^l)\\
    u^{l+1} &= W \left(T^*(Z_0^l+ \Lambda_0^l/\rho) - \frac{\tau }{2\rho} {e}_1\right)
\end{align*}
Here $W$ is the diagonal matrix with entries
\[
	W_{ii} = \begin{cases} 
		\frac{1}{n} & i=1\\
		\frac{1}{2(n-i+1)} & i>1
	\end{cases}
\]
and we introduced the partitions:
\[
	Z^l = \begin{bmatrix} Z_0^l & z_1^l \\ {z_1^l}^* & Z^l_{n+1,n+1} \end{bmatrix} ~~~\mbox{and}~~~
	\Lambda^l = \begin{bmatrix} \Lambda_0^l & \lambda_1^l \\ {\lambda_1^l}^* & \Lambda^l_{n+1,n+1} \end{bmatrix}\,.
\]
The $Z$ update is simply the projection onto the positive definite cone
\begin{equation}\label{eq:z-step}
\hspace{-.3cm}Z^{l+1}:=	 \arg\min_{Z\succeq 0} \left\|Z-\begin{bmatrix}
  T(u^{l+1}) &  x^{l+1} \\
   {x^{l+1}}^* & t^{l+1}
 \end{bmatrix}+\Lambda^{l}/\rho\right\|_F^2\,.
\end{equation}
Projecting a matrix $Q$ onto the positive definite cone is accomplished by
forming an eigenvalue decomposition of $Q$ and setting all negative eigenvalues
to zero.

To summarize, the update for $(t,u,x)$ requires averaging the diagonals of a
matrix (which is equivalent to projecting a matrix onto the space of Toeplitz
matrices), and hence operations that are $O(n)$. The update for $Z$ requires
projecting onto the positive definite cone and requires $O(n^3)$ operations. The update for $\Lambda$ is simply
addition of symmetric matrices.  

Note that the dual solution $\hat{z}$ can be obtained as $\hat{z} = y - \hat{x}$ from the
primal solution $\hat{x}$ obtained from ADMM by using
Lemma~\ref{lem:dual-problem}.

\end{document}